\documentclass[journal,twoside,print]{ieeecolor}
\usepackage{lcsys}
\usepackage{booktabs}
\usepackage{mathtools}
\usepackage{cite}
\usepackage{amsmath,amssymb,amsfonts}
\usepackage{graphicx}   % For including images
\usepackage{caption}    % For customizing captions
\usepackage{subcaption} % For subfigures
\usepackage{url}
\usepackage{amsmath}
\usepackage{amsthm, amssymb}
\usepackage{algorithm}
\usepackage{algpseudocode} % modern
\usepackage{float} % only needed if you use [H] placement

\usepackage{xcolor}   % V4: blue font for newly added content

\newenvironment{list4}{
	\begin{list}{$\bullet$}{%
			\setlength{\itemsep}{0.05cm}
			\setlength{\labelsep}{0.2cm}
			\setlength{\labelwidth}{0.3cm}
			\setlength{\parsep}{0in} 
			\setlength{\parskip}{0in}
			\setlength{\topsep}{0in} 
			\setlength{\partopsep}{0in}
			\setlength{\leftmargin}{0.16in}}}
	{\end{list}}

\theoremstyle{definition}
\newtheorem{definition}{Definition}
\theoremstyle{assumption}
\newtheorem{assumption}{Assumption}
\theoremstyle{problem}

\theoremstyle{lemma}
\newtheorem{lemma}{Lemma}
\theoremstyle{remark}
\newtheorem{remark}{Remark}
\theoremstyle{theorem}
\newtheorem{theorem}{Theorem}
\theoremstyle{corollary}
\newtheorem{corollary}{Corollary}
\newtheorem{proposition}{Proposition}

\usepackage{breqn}
\newcommand{\zono}[1]{\langle #1 \rangle}

\newcommand{\TC}[1]{{\color{blue}#1}}

\usepackage{textcomp}
\def\BibTeX{{\rm B\kern-.05em{\sc i\kern-.025em b}\kern-.08em
    T\kern-.1667em\lower.7ex\hbox{E}\kern-.125emX}}

\begin{document}

\title{Reachability Analysis With Probabilistic Zonotopes: Learning Realized Disturbances and Refining Aleatory Uncertainty}

\author{Amir Modares, Zhen Zhang, Themistoklis Charalambous, Amr Alanwar, Hamidreza Modares
\thanks{A. Modares and T. Charalambous are with the School of Engineering, University of Cyprus, 1678 Nicosia, Cyprus.  E-mails: \texttt{\{modarres.amir,~charalambous.themistoklis\}@ucy.ac.cy}. T. Charalambous is also a Visiting Professor at the School of Electrical Engineering, Aalto University, 02150 Espoo, Finland.}
\thanks{Z. Zhang and A. Alanwar are with Technical University of Munich, Germany. Emails: \texttt{\{zhenzhang.zhang, alanwar\}@tum.de}.} 
\thanks{Hamidreza Modares is with Michigan State University, USA. E-mail: \texttt{ modaresh@msu.edu}.}
}

\maketitle

\begin{abstract}
This paper develops a data-driven reachability framework for linear systems
whose disturbances are modeled by probabilistic zonotopes (PZs), combining
bounded deterministic and Gaussian stochastic components. In contrast to methods
that require a precisely known disturbance model (either purely deterministic or purely
stochastic), we assume only a conservative prior PZ and refine
it from data. The framework separates two uncertainty sources: realized
disturbances, which act along the collected trajectory and govern the size of the
data-consistent model set, and aleatory disturbances, which enter as future
additive uncertainty during reachable-set propagation; both shape the reachable
sets, but through different mechanisms. Refinement exploits prior system knowledge
together with trajectory-consistency constraints induced by the data, which impose
affine couplings between deterministic and Gaussian latent variables. We
accordingly develop a constrained-PZ calculus that absorbs the stochastic part of
these constraints into an equivalent representation, removes infeasible latent
directions, and reduces stochastic covariance, together with identification-aware
fusion rules for combining heterogeneous constrained-PZ descriptions. The refined
realized-disturbance proxies then serve as scenarios in a linear program that
learns the smallest translated and scaled copy of the prior disturbance set that
contains all proxy confidence sets while remaining nested in the prior. The
resulting deterministic, high-probability reachable sets carry formal containment
guarantees with substantially reduced conservatism, and numerical examples confirm
that the pipeline tightens both the data-consistent model set and the propagated
reachable sets.
\end{abstract}

\section{Introduction}
Reachability analysis is a fundamental tool in control theory for characterizing the set of states that a dynamical system can attain under uncertainty~\cite{mitchell2005time}. Accurate reachable-set overapproximations are essential for safety verification~\cite{xue2023reach}, robust control synthesis~\cite{sieber2022system}, and constraint satisfaction~\cite{zhang2022robust}, particularly in safety-critical applications such as autonomous systems~\cite{wang2023safe}, robotics~\cite{xiang2020reachable}, and cyber–physical systems~\cite{chen2022reachability}. The practical usefulness of reachability methods hinges on their ability to balance computational tractability with tightness of the resulting sets~\cite{chen2018decomposition}.

% While there is a considerable amount of literature on computing reachable sets for different model classes, these methods assume an a priori given model. Obtaining a model that adequately describes the system from first principles or noisy data is usually challenging and time-consuming. Simultaneously, system data in the form of measured trajectories are often readily available in many applications. Therefore, we are interested in reachability analysis directly from noisy data of an unknown system model. One recent contribution in this direction can be found in \cite{conf:murat}, where the authors introduce two data-driven methods for computing the reachable sets with probabilistic guarantees. The first method represents the reachability problem as a binary classification problem using a Gaussian process classifier. The second method makes use of a Monte Carlo sampling approach to compute the reachable set. A probabilistic reachability analysis is proposed for general nonlinear systems using level sets of Christoffel functions in \cite{conf:murat_christoffel} where they guarantee that the algorithm's output is an accurate reachable set approximation in a probabilistic sense. 

Despite the extensive progress on reachability computation, the majority of existing approaches are inherently model-based. This includes Hamilton--Jacobi reachability methods and their modern variants \cite{doshi2022hamilton,ganai2024hamilton}, as well as set-propagation approaches based on structured set representations and toolchains \cite{althoff2021set,luo2023reachability}. In practice, however, constructing a sufficiently accurate model from first principles or identifying it reliably from imperfect measurements is often challenging \cite{wang2025system,althoff2023checking}. 
Recent years have seen growing interest in data-driven reachability methods, which leverage measured input–state trajectories to reduce modeling uncertainty when explicit system identification is challenging or unreliable~\cite{alanwar2023data}. Most existing data-driven approaches assume disturbances are bounded within a known deterministic set~\cite{alanwar2023data,hu2025robust}. However, it often leads to excessive conservatism in practice, as the disturbance bounds must accommodate rare or extreme realizations (outliers) observed in data~\cite{gao2022robust}. 

Beyond worst-case bounded-noise models, stochastic reachability studies
high-probability guarantees for systems driven by random uncertainty. This
includes probabilistic reachable and invariant sets for linear systems
\cite{fiacchini2021probabilistic}, chance-constrained stochastic MPC
\cite{hewing2019scenario,hewing2020recursively}, and safety assessment under
state-uncertainty envelopes in autonomous systems
\cite{althoff2010reachability,shetty2020predicting}. Related set-based
estimation methods also provide probabilistic consistency guarantees that
account for epistemic uncertainty together with aleatory disturbances
\cite{li2022set}. These approaches reduce the conservatism of purely
worst-case models while retaining interpretable safety margins.

A complementary direction is to reduce uncertainty itself before propagation. Several works incorporate uncertainty refinement using data
\cite{UncertaintyLearning4}, prior knowledge
\cite{raghuraman2022set,UncertaintyLearning3,UncertaintyLearning1}, or both
\cite{scott2016constrained,UncertaintyLearning2,UncertaintyLearning5}.
However, these refinement methods are typically restricted to deterministic
uncertainty descriptions, most commonly zonotopes
\cite{girard2005reachability} or polytopes~\cite{blanchini1999set}. As a
result, they cannot naturally exploit the stochastic structure needed by
probabilistic reachability. Moreover, existing approaches generally treat uncertainty sources independently~\cite{chen2024robust}, without explicitly distinguishing between past realized disturbances~\cite{der2009aleatory} (which have already occurred in the data and are fixed once observed) and aleatory disturbances~\cite{li2025aleatory}, which represent persistent, inherently stochastic effects during online operation. This lack of separation prevents data from being used in the most effective
way: past realized disturbances should refine the model set, whereas future
aleatory disturbances should refine the disturbance set used in propagation.
With probabilistic zonotopes, however, this separation is technically
nontrivial. Data-consistency and prior-knowledge constraints introduce affine
couplings between deterministic and Gaussian latent variables, so classical zonotope refinements are insufficient. Moreover, combining prior-induced and data-derived descriptions requires fusion rules that preserve probabilistic
confidence while reducing the effective uncertainty.

This paper addresses these challenges by developing a unified probabilistic-zonotope framework for data-driven reachability. The main novelties are fivefold. First, we relax the common assumption that the disturbance model is exactly known and either purely deterministic or purely
stochastic; instead, we start from a conservative prior that combines bounded deterministic and Gaussian stochastic components and refine it from data. Second, we explicitly separate past realized disturbances from future aleatory disturbances, so that past disturbances tighten the data-consistent model set while future disturbances are refined for reachable-set propagation. Third, we develop a probabilistic-zonotope calculus for affine consistency constraints, which absorbs the stochastic part of the constraints, removes infeasible latent directions, and reduces stochastic covariance while retaining the remaining
deterministic equality structure. Fourth, we introduce identification-aware fusion rules that combine heterogeneous constrained probabilistic-zonotope descriptions according to their projected effect on the model set. Finally, building on the refined realized-disturbance proxies, we further contract the conservative aleatory disturbance prior through a linear program that learns the smallest translated and scaled copy of the prior confidence set containing all proxy confidence sets while remaining nested inside the original prior. The combined refinement of both uncertainty sources yields deterministic high-probability reachable sets with formal containment guarantees and significantly lower conservatism than purely deterministic or unrefined stochastic disturbance models.

\section{Notations and Definitions}

\subsection{Notation}

$\mathbb{N}$ denotes the set of natural numbers. $\mathbb{R}^{n}$ denotes the real linear space for all real vectors with dimensions $n \in \mathbb{N}$, and $\mathbb{R}^{m \times n}$ denotes the real linear space for all real matrices with dimensions $m \times n$, $m,n\in\mathbb{N}$. For vectors $x,y \in \mathbb{R}^n$, the notation $x \le y$ (resp.\ $x < y$, $x \ge y$)
is interpreted elementwise, i.e., $x_i \le y_i ~\forall~ i=1,\ldots,n$.
The notation $[-1,1]^{n}$ is the $n$-dimensional hypercube defined as the Cartesian product of the interval $[-1,1]$ with itself $n$ times.
The operator $\mathrm{col}(x_0,\ldots,x_{n-1})$ denotes the vertical stacking of vectors, i.e., $\mathrm{col}(x_0,\ldots,x_{n-1}) = [x_0^\top, \ldots, x_{n-1}^\top]^\top$.
For a matrix $A$, \(A\succ0\) means that
\(A\) is positive definite, $\|A\|_{\infty}$ is its infinity norm, $\mathrm{vec}(A)$ stacks its columns into a single column vector, $A^\perp$ denotes its orthogonal complement,
 $A^\top$ 
 is its transpose, $A^{\dagger}$ is its pseudoinverse,  $\mathrm{tr}(A)$ is its trace, $\mathrm{Im}(A)$ and $\ker(A)$ denote its image (column space)
and nullspace, respectively, and $\mathrm{rank}(A)$ denotes its rank. For matrices $A_i \in \mathbb{R}^{n_i \times m_i}$, $n_i, m_i \in \mathbb{N}$, $i=1,\ldots,n$, 
$\operatorname{blkdiag}(A_1,\ldots,A_n)$ denotes the block-diagonal matrix whose $i$th diagonal block is $A_i$ and whose off-diagonal blocks are zero. $I_n$ denotes the identity matrix of dimension $n$, and \(\mathbf 1_n\in\mathbb R^n\) denotes the column vector whose entries are all equal to one.  For matrices $A$ and $B$, $A \otimes B$ denotes the Kronecker product.
 $\mathbb S_+^{n}$ denotes the cone of real symmetric positive semidefinite
matrices in $\mathbb R^{n\times n}$.
% The sign function is denoted by $\mathrm{sgn}(\cdot)$.
% The support function $h_{\mathcal S}(v)=\sup_{w\in\mathcal S} v^\top w$ of a set $\mathcal S$ satisfies $h_{\mathcal A+c}(v) = v^\top c + h_{\mathcal A}(v)$ for translation, $h_{G\mathcal A}(v) = h_{\mathcal A}(G^\top v)$ for affine transformations, and the monotonicity property $\mathcal A \subseteq \mathcal B \implies h_{\mathcal A}(v) \le h_{\mathcal B}(v)$ for all $v$. 
% Throughout the paper, 
% $\otimes$ is used to express stacked linear mappings arising from
% vectorized system dynamics. 
For sets $A,B \subset \mathbb{R}^n$, the Minkowski sum is denoted by
$A \oplus B = \{a+b : a\in A,\; b\in B\}$, and $\subsetneq
$ denotes strict subset.

{
All random variables/vectors are defined on a probability space $(\Omega,\mathcal{F},\mathbb{P})$.  The probability of an event $\mathcal{E}\subseteq\Omega$ is denoted by $\mathbb{P}(\mathcal{E})$. For random variables $X$ and $Y$, $p(X)$ and $p(X\mid Y)$ denote the probability density or mass function of $X$ and the conditional density or mass function of $X$ given $Y$, respectively. The abbreviation “a.s.” stands for “almost surely,” meaning that the stated property holds with probability one.
The probability of an event 
Expectation, covariance, and cross-covariance are denoted by
$\mathbb{E}[\cdot]$, $\mathrm{Var}(\cdot)$, and $\mathrm{Cov}(\cdot,\cdot)$,
respectively. 
$\mathcal{N}(\mu,\Sigma)$ denotes a Gaussian distribution with mean
$\mu$ and covariance matrix $\Sigma$. Finally, $\operatorname{erf}(\gamma/\sqrt{2})^{n}$ denotes the
probability mass of an $n$-dimensional standard Gaussian contained in the hypercube $[-\gamma,\gamma]^n$.

\subsection{Definitions}
This subsection formally defines the hierarchy of zonotopic sets used to model uncertainty in our framework. 

\begin{definition}[\textbf{Zonotope}{\cite{kuhn1998rigorously}}]
% \textbf{(Zonotope)}\cite{kuhn1998rigorously}
Given a center $c \in \mathbb{R}^{n}$, $n\in \mathbb{N}$, and a generator matrix
$G_d=[g^1,\ \cdots,\ g^{m_d}] \in \mathbb{R}^{n \times {m_d}}$, $m_d\in \mathbb{N}$, the {zonotope} is defined as
\begin{align*}
\mathcal{Z}
= \zono{c,G_d}_{\mathrm{Z}}
&= \Bigl\{
    z = c + G_d\alpha
    :
    \alpha \in [-1,1]^{m_d}
  \Bigr\}.
%   \\
% &= \Bigl\{
%     z = c + \sum_{i=1}^{m_d} \alpha_i g^i
%     ,
%     \alpha_i \in [-1,1]
%   \Bigr\}.
\end{align*}

% \hfill$\square$
\end{definition}

\begin{definition}[\textbf{Gaussian Zonotope} {\cite{althoff2009safety}}]
\label{DEf:GZ} 
Given a center $c \in \mathbb{R}^{n}$, and a generator matrix 
$G_s=[g^1, \cdots, g^{m_s}] \in \mathbb{R}^{n \times m_s}$, $m_s\in \mathbb{N}$,
the Gaussian zonotope is defined as
\begin{align*}
   \mathcal{Z}_{\mathcal{N}} = \zono{c,G_s}_{\mathcal{N}}
    = \Bigl\{z= c + G_s \nu: \nu \sim \mathcal{N}(0,I_{m_s}) \Bigl\},
\end{align*}
where $z \sim \mathcal{N}(c,\, G_sG_s^{\top})$.
\end{definition}

\begin{remark}
Since
$G_sG_s^\top$ is symmetric positive semidefinite, the Gaussian random vector $z\sim\mathcal{N}(c,G_sG_s^\top)$ is well defined. If $G_s$ has full row rank, then $G_sG_s^\top\succ0$ and the distribution is nondegenerate. Otherwise, it is a degenerate Gaussian supported on $c+\operatorname{range}(G_s)$. 
\end{remark}

\begin{definition}[\textbf{Probabilistic Zonotope (PZ)}{\cite{althoff2009safety}}]
\label{def:mixed_zonotope}
Given $c \in\mathbb{R}^n$, $G_d\in\mathbb{R}^{n\times m_d}$, and 
$G_s\in\mathbb{R}^{n\times m_s}$, the PZ is defined as
\begin{align}
\label{eq:mixed_zono_def}
&\mathcal{Z}_{\mathrm{P}} = \zono{c,G_d,G_s}_{\mathrm{P}}
\\ &
= \Bigl\{
z = c + G_d\alpha + G_s\nu :
\alpha\in[-1,1]^{m_d},\;
\nu\sim\mathcal{N}(0,I_{m_s})
\Bigr\}. \nonumber
\end{align}
\end{definition}

% \begin{remark}
% \label{rem:mixed_sum_zero_bias}
% A probabilistic zonotope $\mathcal{Z}_{\mathcal{P}}$ can be represented as the Minkowski sum of a deterministic
% zonotope $\mathcal{Z}=\zono{c_d,G_d}_{\mathcal{Z}}$ and a Gaussian zonotope $\mathcal{Z}_{\mathcal N}=\zono{c_s,G_s,\mu_s,\Lambda_s}_{\mathcal N}$, i.e.,
% \(
% \mathcal{Z}_{\mathcal{P}}=\mathcal{Z}\ \oplus\ \mathcal{Z}_{\mathcal N}=\zono{c,G_d,G_s,\mu_s,\Lambda_s}_{\mathcal{P}},
% % =
% % \zono{c_d,G_d}_{\mathcal{Z}}
% % \ \oplus\
% % \zono{c_s,G_s,\mu_s,\Lambda_s}_{\mathcal N},
% \)
% where the overall center satisfies $c = c_d + c_s$.
% Since the constant term can be decomposed into a deterministic center and a
% stochastic bias in multiple ways, this decomposition is not unique.
% Nevertheless, all such decompositions represent the same disturbance
% distribution and are therefore equivalent.

% \end{remark}

\begin{definition}[\textbf{Constrained Probabilistic Zonotope (CPZ)}]
\label{def:cmixed_zonotope}
Let $c \in\mathbb{R}^n$, $G_d\in\mathbb{R}^{n\times m_d}$,
$G_s\in\mathbb{R}^{n\times m_s}$,
$A_d \in \mathbb{R}^{p \times m_d}$,
$A_s \in \mathbb{R}^{p \times m_s}$, and
$b \in \mathbb{R}^{p}$, $p\in\mathbb{N}$.
The CPZ is defined as
\begin{align}
\label{eq:cmz}
\mathcal{Z}_{\mathrm{CP}}
&= \zono{c,G_d,G_s,A_d,A_s}_{\mathrm{CP}} \nonumber\\
&= \Bigl\{
z = c + G_d\alpha + G_s\nu
\; : \;
A_d\alpha + A_s\nu = b,
\nonumber\\[-1mm]
&\hspace{4.2em}
\alpha \in [-1,1]^{m_d},\;
\nu \sim \mathcal{N}(0,I_{m_s})
\Bigr\}.
\end{align}
\end{definition}

We now propose matrix probabilistic zonotopes as follows.

% \begin{definition}[\textbf{Matrix Probabilistic Zonotope (MPZ)}]
% \label{def:matrix_mixed_zono}
% Let \(C\in\mathbb R^{n\times N}\), and let
% \(
% \widetilde G_d
% =
% \bigl[
% G_d^{(1)}\;\cdots\;G_d^{(m_d)}
% \bigr],
% \,\,
% \widetilde G_s
% =
% \bigl[
% G_s^{(1)}\;\cdots\;G_s^{(m_s)}
% \bigr],
% \)
% where \(G_d^{(i)},G_s^{(j)}\in\mathbb R^{n\times N}\) are matrix
% generators. The matrix probabilistic zonotope is defined as
% \begin{align}
% \label{eq:W_matrix_mixed}
% \mathcal Z_{\mathcal{MP}}
% &=
% \zono{C,\widetilde G_d,\widetilde G_s}_{\mathcal{MP}}
% \nonumber\\
% &=
% \Bigl\{
% Z=C+
% \sum_{i=1}^{m_d}\alpha^{(i)}G_d^{(i)}
% +
% \sum_{j=1}^{m_s}\nu^{(j)}G_s^{(j)}
% \; ; \;
% \nonumber\\
% &\hspace{1.4em}
% \alpha\in[-1,1]^{m_d},\;
% \nu\sim\mathcal N(0,I_{m_s})
% \Bigr\}.
% \end{align}
% It is a 
% \end{definition}

\begin{definition}[\textbf{Constrained Matrix Probabilistic Zonotopes (CMPZ)}]
\label{def:matrix_and_constrained_mixed_zono}
Let \(C\in\mathbb R^{n\times N}\), $N\in\mathbb{N}$, and let
\(
G_d
=
\bigl[
G_d^{(1)}\;\cdots\;G_d^{(m_d)}
\bigr],
\,
G_s
=
\bigl[
G_s^{(1)}\;\cdots\;G_s^{(m_s)}
\bigr],
\)
where \(G_d^{(i)},G_s^{(j)}\in\mathbb R^{n\times N}\) are matrix
generators. \(
A_d
=
\bigl[
A_d^{(1)}\;\cdots\;A_d^{(m_d)}
\bigr],
\,
A_s
=
\bigl[
A_s^{(1)}\;\cdots\;A_s^{(m_s)}
\bigr],
\)
with \(A_d^{(i)},A_s^{(j)},B\in\mathbb R^{p\times q}\). Then, the constrained matrix probabilistic zonotope (CMPZ) is defined
as
\begin{align}
\label{eq:CMPZ_def}
\mathcal Z_{\mathrm{CMP}}
&=
\langle
C,G_d,G_s,A_d,A_s,B
\rangle_{\mathrm{CMP}}
\nonumber\\
&=
\Bigl\{
Z=C+
\sum_{i=1}^{m_d}\alpha^{(i)}G_d^{(i)}
+
\sum_{j=1}^{m_s}\nu^{(j)}G_s^{(j)}
\; : \;
\nonumber\\[-1mm]
&\hspace{2.5em}
\sum_{i=1}^{m_d}\alpha^{(i)}A_d^{(i)}
+
\sum_{j=1}^{m_s}\nu^{(j)}A_s^{(j)}
=
B,\;
\nonumber\\[-1mm]
&\hspace{2.5em}
\alpha\in[-1,1]^{m_d},\;
\nu\sim\mathcal N(0,I_{m_s})
\Bigr\}.
\end{align}
In the absence of the affine matrix equality constraint on the latent variables, the CMPZ reduces to the matrix probabilistic zonotope (MPZ).
% As the unconstrained special case of the CMPZ, obtained by removing the affine matrix equality constraint imposed on the latent variables, the matrix probabilistic zonotope (MPZ) is defined as
% \begin{align}
% \label{eq:W_matrix_mixed}
% \mathcal Z_{\mathcal{MP}}
% &=
% \langle C,\, G_d,\, G_s \rangle_{\mathcal{MP}}
% \nonumber\\
% &=
% \Bigl\{
% Y=C+
% \sum_{i=1}^{m_d}\alpha^{(i)}G_d^{(i)}
% +
% \sum_{j=1}^{m_s}\nu^{(j)}G_s^{(j)}
% \; : \;
% \nonumber\\[-1mm]
% &\hspace{1.4em}
% \alpha\in[-1,1]^{m_d},\;
% \nu\sim\mathcal N(0,I_{m_s})
% \Bigr\}.
% \end{align}

% If, in addition, the same latent variables are required to satisfy an affine
% matrix equality, then the constrained matrix probabilistic zonotope (CMPZ) is defined
% as
% \begin{align}
% \label{eq:CMPZ_def}
% \mathcal Z_{\mathcal{CMP}}
% &=
% \langle
% C,G_d,G_s,
% A_d,A_s,B
% \rangle_{\mathcal{CMP}}
% \nonumber\\
% &=
% \Bigl\{
% Y\in\mathcal Z_{\mathcal{MP}}
% \; ; \;
% \sum_{i=1}^{m_d}\alpha^{(i)}A_d^{(i)}
% +
% \sum_{j=1}^{m_s}\nu^{(j)}A_s^{(j)}
% =
% B
% \Bigr\},
% \end{align}
% where
% \(
% A_d
% =
% \bigl[
% A_d^{(1)}\;\cdots\;A_d^{(m_d)}
% \bigr],
% \,
% A_s
% =
% \bigl[
% A_s^{(1)}\;\cdots\;A_s^{(m_s)}
% \bigr],
% \)
% with \(A_d^{(i)},A_s^{(j)},B\in\mathbb R^{p\times q}\). 
\end{definition}}

%  \begin{definition}[\textbf{Matrix and Constrained Matrix Probabilistic Zonotopes}]
% \label{def:matrix_and_constrained_mixed_zono}
% Let \(C\in\mathbb R^{n\times N}\), and let
% \(
% G_d
% =
% \bigl[
% G_d^{(1)}\;\cdots\;G_d^{(m_d)}
% \bigr],
% \,
% G_s
% =
% \bigl[
% G_s^{(1)}\;\cdots\;G_s^{(m_s)}
% \bigr],
% \)
% where \(G_d^{(i)},G_s^{(j)}\in\mathbb R^{n\times N}\) are matrix
% generators. The matrix probabilistic zonotope (MPZ) is defined as
% \begin{align}
% \label{eq:W_matrix_mixed}
% \mathcal Z_{\mathcal{MP}}
% &=
% \langle C,\, G_d,\, G_s \rangle_{\mathcal{MP}}
% \nonumber\\
% &=
% \Bigl\{
% Y=C+
% \sum_{i=1}^{m_d}\alpha^{(i)}G_d^{(i)}
% +
% \sum_{j=1}^{m_s}\nu^{(j)}G_s^{(j)}
% \; : \;
% \nonumber\\[-1mm]
% &\hspace{1.4em}
% \alpha\in[-1,1]^{m_d},\;
% \nu\sim\mathcal N(0,I_{m_s})
% \Bigr\}.
% \end{align}

% If, in addition, the same latent variables are required to satisfy an affine
% matrix equality, then the constrained matrix probabilistic zonotope (CMPZ) is defined
% as
% \begin{align}
% \label{eq:CMPZ_def}
% \mathcal Z_{\mathcal{CMP}}
% &=
% \langle
% C,G_d,G_s,
% A_d,A_s,B
% \rangle_{\mathcal{CMP}}
% \nonumber\\
% &=
% \Bigl\{
% Y\in\mathcal Z_{\mathcal{MP}}
% \; ; \;
% \sum_{i=1}^{m_d}\alpha^{(i)}A_d^{(i)}
% +
% \sum_{j=1}^{m_s}\nu^{(j)}A_s^{(j)}
% =
% B
% \Bigr\},
% \end{align}
% where
% \(
% A_d
% =
% \bigl[
% A_d^{(1)}\;\cdots\;A_d^{(m_d)}
% \bigr],
% \,
% A_s
% =
% \bigl[
% A_s^{(1)}\;\cdots\;A_s^{(m_s)}
% \bigr],
% \)
% with \(A_d^{(i)},A_s^{(j)},B\in\mathbb R^{p\times q}\). 
% \end{definition}}

\begin{remark}
\label{CMPZ-CPZ}
Consider the CMPZ in~\eqref{eq:CMPZ_def}. By vectorizing the matrix variable
\(Z\), the CMPZ becomes the CPZ
\begin{align}
\operatorname{vec}(\mathcal Z_{\mathrm{CMP}})
=
\Bigl\{
&z=c_v+G_{d,v}\alpha+G_{s,v}\nu
\; ; \;
\nonumber\\[-1mm]
&A_{d,v}\alpha+A_{s,v}\nu=b_v,\quad
\alpha\in[-1,1]^{m_d},
\nonumber\\[-1mm]
&\nu\sim\mathcal N(0,I_{m_s})
\Bigr\},
\end{align}
where \vspace{-12pt}
\begin{align*}
c_v&:=\operatorname{vec}(C),\qquad
b_v:=\operatorname{vec}(B),\\
G_{d,v}
&:=
[\operatorname{vec}(G_d^{(1)})\;\cdots\;
\operatorname{vec}(G_d^{(m_d)})],\\
G_{s,v}
&:=
[\operatorname{vec}(G_s^{(1)})\;\cdots\;
\operatorname{vec}(G_s^{(m_s)})],\\
A_{d,v}
&:=
[\operatorname{vec}(A_d^{(1)})\;\cdots\;
\operatorname{vec}(A_d^{(m_d)})],\\
A_{s,v}
&:=
[\operatorname{vec}(A_s^{(1)})\;\cdots\;
\operatorname{vec}(A_s^{(m_s)})].
\end{align*}
Thus, a CMPZ is exactly a CPZ in the vectorized matrix space
\(\mathbb R^{nN}\). Moreover, the \(k\)-th column is obtained by the linear projection
$E_k:=e_k^\top\otimes I_n, \, Z_{(:,k)}=E_k\operatorname{vec}()$, where \(e_k\) is the \(k\)-th standard basis vector in \(\mathbb R^N\).
Since the projection is applied to the constrained set, the equality
constraint is preserved. Hence, the projected column set is $\mathcal Z_{\mathrm{CP},k}
=
\{E_kz:\; z\in\operatorname{vec}(\mathcal Z_{\mathrm{CMP}})\}$.
 Therefore, the column projection does not remove the global equality
constraint. In general, all projected column sets share the same latent
variables \((\alpha,\nu)\) and the same equality constraint
\(
A_{d,v}\alpha+A_{s,v}\nu=b_v .
\)
Thus, the CMPZ represents a coupled collection of column projections, not a
Cartesian product of independent column CPZs, unless the generators and
constraints are block-separated across columns.
\end{remark}

Inspired by the probabilistic bounding techniques introduced in~\cite{althoff2009safety}, we evaluate the probability metrics of standard Gaussian vectors before generalizing to set-based structures. From~\cite[Proposition~2]{althoff2009safety}, for an $n$-dimensional independent standard Gaussian vector $\nu \sim \mathcal{N}(0, I_n)$, the probability that all components remain bounded within $[-\gamma_\delta, \gamma_\delta]^n$ is $\mathbb{P}(\nu \in [-\gamma_\delta, \gamma_\delta]^n) = \operatorname{erf}(\gamma_\delta/\sqrt{2})^n$. For a matrix probabilistic zonotope, by Remark~\ref{CMPZ-CPZ} the stochastic object lives in the vectorized space $\mathbb{R}^{nN}$, so the corresponding probability becomes $\operatorname{erf}(\gamma_\delta/\sqrt{2})^{nN}$. Enforcing these probabilities to meet a target confidence level $1-\delta$ leads to the following unified confidence operator for PZs and MPZs.

\begin{definition}[\textbf{High-Probability Zonotopes via $\delta$-Confidence Sets}]
\label{def:confidence_set}
Let $\delta\in(0,1)$ be a chosen confidence level. For a probabilistic element residing in an $n$-dimensional space, define the truncation width
\begin{equation}
\gamma_\delta := \sqrt{2}\,\operatorname{erf}^{-1}\!\bigl((1-\delta)^{1/n}\bigr).
\label{eq:trunc_width}
\end{equation}
Then, the $\delta$-confidence sets are constructed as follows:

\noindent\textbf{(i) PZ:} For a probabilistic zonotope $\mathcal{Z}_{\mathrm{P}} = \langle c,\, G_d,\, G_s \rangle_{\mathrm{P}}$ defined in \eqref{eq:mixed_zono_def}, its $\delta$-confidence set is
\begin{equation}
\label{eq:PZ_confidence_set}
\begin{split}
 \mathcal{Z}_{\mathrm{P}}(\delta) = \Bigl\{z = c &+ \sum_{i=1}^{m_d}\alpha^{(i)}g_d^{(i)} + \gamma_\delta\sum_{j=1}^{m_s}\beta^{(j)}g_s^{(j)} \;;\; \\
 &\alpha\in[-1,1]^{m_d},\; \beta\in[-1,1]^{m_s}\Bigr\}.
\end{split}
\end{equation}

\noindent\textbf{(ii) MPZ:} For a matrix probabilistic zonotope $\mathcal{Z}_{\mathrm{MP}} = \langle C,\, G_d,\, G_s \rangle_{\mathrm{MP}}$ defined in \eqref{eq:CMPZ_def}, the stochastic object lives in $\mathbb{R}^{nN}$, so the truncation width is $\gamma_\delta^{\mathrm{MPZ}} := \sqrt{2}\,\operatorname{erf}^{-1}((1-\delta)^{1/(nN)})$, and its $\delta$-confidence set is
\begin{equation}
\label{eq:MPZ_confidence_set}
\begin{split}
\mathcal{Z}_{\mathrm{MP}}(\delta) = \Bigl\{Z = C &+ \sum_{i=1}^{m_d}\alpha^{(i)}G_d^{(i)} + \gamma_\delta^{\mathrm{MPZ}}\sum_{j=1}^{m_s}\beta^{(j)}G_s^{(j)} \;;\; \\
&\alpha\in[-1,1]^{m_d},\; \beta\in[-1,1]^{m_s}\Bigr\}.
\end{split}
\end{equation}

In both cases, the confidence set is deterministic and bounded and one has
\begin{align*}
 \mathbb{P}\!\left(\mathcal{Z}_{\mathrm{P}}\in\mathcal{Z}_{\mathrm{P}}(\delta)\right)
 \geq
 \operatorname{erf}\!\left(\frac{\gamma_\delta}{\sqrt{2}}\right)^{\!n}
 = 1-\delta,
 \\
 \mathbb{P}\!\left(\mathcal{Z}_{\mathrm{MP}}\in\mathcal{Z}_{\mathrm{MP}}(\delta)\right)
 \geq
 \operatorname{erf}\!\left(\frac{\gamma_\delta^{\mathrm{MPZ}}}{\sqrt{2}}\right)^{\!nN}
 = 1-\delta.
 \label{eq:confidence_prob}
\end{align*}
For constrained PZs or CMPZs with deterministic equality constraints, such as $A_d\alpha=b$, the $\delta$-confidence operation leaves the equality unchanged; it only bounds the Gaussian factor and scales the stochastic generators by the respective $\gamma_\delta$.
\end{definition}

% \begin{remark}
% The \(\delta\)-confidence set replaces the Gaussian factor
% \(\nu\sim\mathcal N(0,I_{m_s})\) with a bounded factor
% \(\beta\in[-1,1]^{m_s}\) and scales the stochastic generator as \(\delta G_s\). Thus, larger \(\delta\) gives higher confidence but a larger set, while smaller \(\delta\) gives a tighter set with lower confidence. 
% % A PZ or MPZ confidence set is therefore a deterministic zonotope or matrix
% % zonotope with augmented generator \([G_d\;\;\delta G_s]\). Since CPZs are
% % converted to PZs by Lemma~\ref{lem:generalized_constraint_absorption}, their truncation is handled through the resulting PZ representation.
% \end{remark}

The following lemma is required.

\begin{lemma}[\hspace{0.00001cm}{\cite{anderson2005optimal}}]
\label{lem:gaussian_conditioning}
Let
\(x=[x_\alpha^\top\;x_\beta^\top]^\top\sim\mathcal N(\mu,\Sigma)\),
where
\(\mu=[\mu_\alpha^\top\;\mu_\beta^\top]^\top\) and
\[
\Sigma=
\begin{bmatrix}
\Sigma_{\alpha\alpha} & \Sigma_{\alpha\beta}\\
\Sigma_{\alpha\beta}^{\top} & \Sigma_{\beta\beta}
\end{bmatrix},
\qquad
\Sigma_{\beta\beta}\succ0 .
\]
Then,
\(x\mid x_\beta=\bar\beta
\sim\mathcal N(\mu_{\mathrm{cond}},\Sigma_{\mathrm{cond}})\),
where
\[
\mu_{\mathrm{cond}}
=
\begin{bmatrix}
\mu_{\mathrm{cond},\alpha}\\
\bar\beta
\end{bmatrix},
\qquad
\Sigma_{\mathrm{cond}}
=
\operatorname{blkdiag}(\Sigma_{\mathrm{cond},\alpha\alpha},0),
\]
with
\(
\mu_{\mathrm{cond},\alpha}
=
\mu_\alpha+\Sigma_{\alpha\beta}\Sigma_{\beta\beta}^{-1}
(\bar\beta-\mu_\beta),
\,
\Sigma_{\mathrm{cond},\alpha\alpha}
=
\Sigma_{\alpha\alpha}
-\Sigma_{\alpha\beta}\Sigma_{\beta\beta}^{-1}
\Sigma_{\alpha\beta}^{\top}.
\)
\end{lemma}

\section{Affine Transformations and Information Fusion of Probabilistic Zonotope }

% In the subsequent uncertainty-reduction steps, we repeatedly impose affine equalities and combine multiple probabilistic set descriptions while keeping computations tractable. 
% This section establishes the corresponding calculus for probabilistic zonotopes: how affine constraints can be absorbed into an equivalent unconstrained representation, and how such representations can be fused to contract uncertainty.

% {This section develops results for:
% 1) transforming a CPZ into an equivalent PZ
% by incorporating affine constraints (Lemma~\ref{lem:joint_eq_mixed_complete}), resulting in a Gaussian component with
% reduced covariance; and
% 2) integrating information from a PZ and a CPZ, yielding a PZ with reduced overall uncertainty (Lemma~\ref{lem:support_trunc_domination}).
%  These results will be used in the subsequent sections to integrate prior
% knowledge on the system parameters (Assumption~\ref{PK}) with data and to refine disturbance descriptions by removing
% realizations that are inconsistent with either the data or the prior knowledge.
% The resulting refined disturbance sets are then leveraged in the main results section to obtain tighter
% reachable sets.}

 This section establishes the 
corresponding calculus for uncertainty reduction by developing results for: 1) transforming a 
CPZ into an equivalent CPZ with no stochastic element in the equality constraint (Lemma~\ref{lem:generalized_constraint_absorption}), yielding a representation 
with reduced stochastic covariance and shrunk deterministic generator 
range (Lemma \ref{lem:uncertainty_reduction}); and 3) the fusion of multiple CPZs 
becomes a CPZ with  reduced overall uncertainty (Lemma~\ref{lem:exact_cpz_intersection}). Together, these results  provide the algebraic foundation in the subsequent sections to progressively refine  disturbance sets, ultimately 
yielding tighter reachable-set overapproximations.

\begin{lemma}
\label{lem:generalized_constraint_absorption}
Consider the CPZ
\begin{align}
\mathcal Z_{\mathrm {CP}}
=
\Bigl\{
&z=c+G_d\alpha+G_s\nu, \, A_d\alpha+A_s\nu=b:
\nonumber\\[-1mm] \, \,
&\alpha\in[-1,1]^{m_d}, \, \nu\sim\mathcal N(0,I_{m_s})
\Bigr\},
\end{align}
where \(c\in\mathbb R^n\), \(G_d\in\mathbb R^{n\times m_d}\),
\(G_s\in\mathbb R^{n\times m_s}\), \(A_d\in\mathbb R^{p\times m_d}\),
\(A_s\in\mathbb R^{p\times m_s}\), and \(b\in\mathbb R^p\). Let
\(\operatorname{rank}(A_s)=r\le p<m_s\), and let
\(U=[U_r\;\;U_0]\in\mathbb R^{p\times p}\) be orthogonal, where
\(U_r\in\mathbb R^{p\times r}\) spans \(\operatorname{Im}(A_s)\) and
\(U_0\in\mathbb R^{p\times(p-r)}\) spans \(\ker(A_s^\top)\). Let \(K\in\mathbb R^{m_s\times(m_s-r)}\) be an orthonormal basis for
\(\ker(A_s)\). Define
\begin{align}
\tilde A_d:=U_0^\top A_d,
\qquad
\tilde b:=U_0^\top b .
\label{eq:tilde_Ad_b_def}
\end{align}
If \(\mathcal Z_{\mathrm{CP}}\) is nonempty, then
\(\mathcal Z_{\mathrm{CP}}=\widehat{\mathcal Z}_{\mathrm{CP}}\), where
\begin{align}
\widehat{\mathcal Z}_{\mathrm{CP}}
=
\Bigl\{
&z=c_c+G_{c,d}\alpha+G_{c,s}\zeta, \,\, \tilde A_d\alpha=\tilde b,
\nonumber\\[-1mm]
&\alpha\in[-1,1]^{m_d},\;
\zeta\sim\mathcal N(0,I_{m_s-r})
\Bigr\},
\label{eq:absorbed_cpz_representation}
\end{align}
with $c_c:=c+G_sA_s^{\dagger}b$, $G_{c,d}:=G_d-G_sA_s^{\dagger}A_d$, and
$G_{c,s}:=G_sK\in\mathbb R^{n\times(m_s-r)}$.
\end{lemma}

\begin{proof}
For fixed \(\alpha\), the affine constraint
\(A_d\alpha+A_s\nu=b\) requires
\(
A_s\nu=b-A_d\alpha=:y(\alpha).
\label{eq:y_alpha_def}
\)
Pre-multiplying by \(U^\top=[U_r^\top;\;U_0^\top]\) gives
\begin{equation}
\begin{bmatrix}
U_r^\top A_s\\
U_0^\top A_s
\end{bmatrix}\nu
=
\begin{bmatrix}
U_r^\top y(\alpha)\\
U_0^\top y(\alpha)
\end{bmatrix}.
\end{equation}
Since the columns of \(U_0\) span \(\ker(A_s^\top)\), one has
\(U_0^\top A_s=0\). Hence, the lower block imposes the purely deterministic
condition
\(
U_0^\top(b-A_d\alpha)=0,
\)
or equivalently
\begin{align}
\tilde A_d\alpha=\tilde b .
\label{eq:det_residual_constraint}
\end{align}
Thus, if \(\alpha\) violates \eqref{eq:det_residual_constraint}, then no
\(\nu\) can satisfy the original affine constraint. For every
\(\alpha\) satisfying \eqref{eq:det_residual_constraint}, only the upper
block remains
\begin{align}
\widetilde A_s\nu
:=
(U_r^\top A_s)\nu
=
U_r^\top y(\alpha)
=: \tilde y(\alpha),
\label{eq:upper_block}
\end{align}
where \(\widetilde A_s\in\mathbb R^{r\times m_s}\) is full row rank. Since \(\ker(\widetilde A_s)=\ker(A_s)\), the basis \(K\) satisfies
\(\widetilde A_sK=0\). Define
$
H:=[K\;\;\widetilde A_s^\top]\in\mathbb R^{m_s\times m_s}$.
The columns of \(K\) span \(\ker(\widetilde A_s)\), while the columns of
\(\widetilde A_s^\top\) span \(\operatorname{Im}(\widetilde A_s^\top)\);
these subspaces are orthogonal complements in \(\mathbb R^{m_s}\). Hence,
\(H\) is invertible. Let \(q=H^{-1}\nu\), where
\begin{align*}
H^{-1}
=
\begin{bmatrix}
K^\top\\
(\widetilde A_s\widetilde A_s^\top)^{-1}\widetilde A_s
\end{bmatrix}.
\end{align*}
Partition \(q^\top=[q_\alpha^\top\;q_\beta^\top]\), with
\(q_\alpha\in\mathbb R^{m_s-r}\) and \(q_\beta\in\mathbb R^r\). Since
\(\nu\sim\mathcal N(0,I_{m_s})\), we have
\(q\sim\mathcal N(0,\Sigma_q)\), where
\begin{align}
\Sigma_q
&=
H^{-1}(H^{-1})^\top
=
\operatorname{blkdiag}
\bigl(\Sigma_{q,\alpha\alpha},\Sigma_{q,\beta\beta}\bigr),
\nonumber\\
\Sigma_{q,\alpha\alpha}
&=
K^\top K
=
I_{m_s-r},
\qquad
\Sigma_{q,\beta\beta}
=
(\widetilde A_s\widetilde A_s^\top)^{-1}
.
\end{align}
Therefore, \(q_\alpha\) and \(q_\beta\) are \emph{a priori} independent. Substituting \(\nu=Hq=[K\;\;\widetilde A_s^\top]q\) into
\eqref{eq:upper_block} gives
$[\,0\;\;\widetilde A_s\widetilde A_s^\top\,]
\begin{bmatrix}
q_\alpha\\
q_\beta
\end{bmatrix}
=
\tilde y(\alpha)$,
and therefore $q_\beta
=
(\widetilde A_s\widetilde A_s^\top)^{-1}\tilde y(\alpha)$.
% Thus, the range-space coordinate is fixed by the equality constraint, while the null-space coordinate remains free.

Applying Lemma~\ref{lem:gaussian_conditioning} to
\(q=[q_\alpha^\top\;q_\beta^\top]^\top\), conditioned on
\(
\bar\beta:=
(\widetilde A_s\widetilde A_s^\top)^{-1}\tilde y(\alpha),
\)
yields
\(
\mu_{\mathrm{cond},\alpha}
=
0+
\Sigma_{q,\beta\beta}^{-1}\bar\beta
=
0, \,
\Sigma_{\mathrm{cond},\alpha\alpha}
=
I_{m_s-r}
-
\Sigma_{q,\beta\beta}^{-1}
=
I_{m_s-r}
\), since $\Sigma_{q,\alpha\beta}
=
K^\top\widetilde A_s^\top
(\widetilde A_s\widetilde A_s^\top)^{-1}
=
0$.
Hence,
\begin{align}
q\mid q_\beta=\bar\beta
\sim
\mathcal N(\mu_{\mathrm{cond}},\Sigma_{\mathrm{cond}}),
\end{align}
where
\begin{align*}
\mu_{\mathrm{cond}}
=
\begin{bmatrix}
0\\
(\widetilde A_s\widetilde A_s^\top)^{-1}\tilde y(\alpha)
\end{bmatrix},
\,\,
\Sigma_{\mathrm{cond}}
=
\operatorname{blkdiag}(I_{m_s-r},0).
\end{align*}
Since \(\nu=Hq\), one has
\begin{align}
\mathbb E[\nu\mid q_\beta]
&=
\widetilde A_s^\top
(\widetilde A_s\widetilde A_s^\top)^{-1}\tilde y(\alpha)
=
A_s^{\dagger}y(\alpha),
\label{eq:nu_mean}
\\
\operatorname{Var}(\nu\mid q_\beta)
&=
KK^\top .
\label{eq:nu_var}
\end{align}
Therefore,
\(
\nu\mid A_s\nu=y(\alpha)
\sim
\mathcal N\bigl(A_s^{\dagger}(b-A_d\alpha),KK^\top\bigr).
\)
Equivalently, with
\(\zeta:=q_\alpha\sim\mathcal N(0,I_{m_s-r})\), one may write
\(
\nu=A_s^+(b-A_d\alpha)+K\zeta .
\) Substituting this into
\(z=c+G_d\alpha+G_s\nu\) gives
\begin{align}
z
&=
c+G_d\alpha
+
G_sA_s^{\dagger}(b-A_d\alpha)
+
G_sK\zeta
\nonumber\\
&=
\underbrace{(c+G_sA_s^{\dagger}b)}_{c_c}
+
\underbrace{(G_d-G_sA_s^{\dagger}A_d)}_{G_{c,d}}\alpha
+
\underbrace{G_sK}_{G_{c,s}}\zeta .
\label{eq:z_decomposed}
\end{align}
Together, with the residual deterministic constraint
\(\tilde A_d\alpha=\tilde b\), this shows that every element of
\(\mathcal Z_{\mathrm{CP}}\) belongs to
\(\widehat{\mathcal Z}_{\mathrm{CP}}\).

Conversely, take any \(z\in\widehat{\mathcal Z}_{\mathrm{CP}}\). Then
there exist \(\alpha\in[-1,1]^{m_d}\), satisfying
\(\tilde A_d\alpha=\tilde b\), and
\(\zeta\sim\mathcal N(0,I_{m_s-r})\), such that
\(
z=c_c+G_{c,d}\alpha+G_{c,s}\zeta .
\)
Define
\(
\nu:=A_s^{\dagger}(b-A_d\alpha)+K\zeta .
\)
Since \(\tilde A_d\alpha=\tilde b\), one has
\(U_0^\top(b-A_d\alpha)=0\), which implies
\(b-A_d\alpha\in\operatorname{Im}(A_s)\). Hence,
\(
A_sA_s^{\dagger}(b-A_d\alpha)=b-A_d\alpha .
\)
Also \(A_sK=0\). Therefore,
\(
A_s\nu
=
A_sA_s^{\dagger}(b-A_d\alpha)+A_sK\zeta
=
b-A_d\alpha,
\)
and consequently \(A_d\alpha+A_s\nu=b\). Substituting this \(\nu\) into the
original representation gives
\(
c+G_d\alpha+G_s\nu
=
c_c+G_{c,d}\alpha+G_{c,s}\zeta
=
z .
\)
Thus \(z\in\mathcal Z_{\mathrm{CP}}\). Hence
\(\widehat{\mathcal Z}_{\mathrm{CP}}\subseteq\mathcal Z_{\mathrm{CP}}\),
and combined with the first inclusion,
\(
\mathcal Z_{\mathrm{CP}}=\widehat{\mathcal Z}_{\mathrm{CP}}.
\)
\end{proof}

\begin{lemma}
\label{lem:uncertainty_reduction}
Let
\begin{align}
\mathcal Z_{\mathrm P}
=
\Bigl\{
&z=c+G_d\alpha+G_s\nu,
\nonumber\\[-1mm]
&\alpha\in[-1,1]^{m_d},\;
\nu\sim\mathcal N(0,I_{m_s})
\Bigr\},
\label{eq:unconstrained_pz_for_reduction}
\end{align}
and define the corresponding CPZ
\begin{align}
\mathcal Z_{\mathrm{CP}}
=
\Bigl\{
z\in\mathcal Z_{\mathrm P}
\; ;\;
A_d\alpha+A_s\nu=b
\Bigr\}.
\label{eq:constrained_pz_for_reduction}
\end{align}
Let \(\widehat{\mathcal Z}_{\mathrm{CP}}\) be the representation of \(\mathcal Z_{\mathrm{CP}}\) via
Lemma~\ref{lem:generalized_constraint_absorption}. Then,
\begin{align}
\widehat{\mathcal Z}_{\mathrm{CP}}
=
\mathcal Z_{\mathrm{CP}}
\subseteq
\mathcal Z_{\mathrm P}.
\label{eq:cpz_subset_original_pz}
\end{align}
Moreover, the stochastic covariance is reduced as
\begin{align}
\Sigma_{\mathrm{orig}}-\Sigma_c
=
G_sP_sG_s^\top\succeq0,
\label{eq:covariance_reduction_revised}
\end{align}
where
\(\Sigma_{\mathrm{orig}}\)
is the original stochastic covariance of $\mathcal Z_{\mathrm P}$,
\(\Sigma_c\)
is the covariance after constraint absorption of $\widehat{\mathcal Z}_{\mathrm{CP}}$,
\(P_s:=I_{m_s}-NN^\top\), and \(\operatorname{rank}(P_s)=r\).
\end{lemma}

\begin{proof}
By construction, \(\mathcal Z_{\mathrm{CP}}\) is obtained from
\(\mathcal Z_{\mathrm P}\) by imposing the additional affine equality
\(A_d\alpha+A_s\nu=b\) on the same latent variables. Therefore, every element
of \(\mathcal Z_{\mathrm{CP}}\) is also an element of
\(\mathcal Z_{\mathrm P}\), and hence
\(
\mathcal Z_{\mathrm{CP}}\subseteq\mathcal Z_{\mathrm P}.
\)
Lemma~\ref{lem:generalized_constraint_absorption} gives an equivalent
representation of \(\mathcal Z_{\mathcal{CP}}\), namely
\(
\widehat{\mathcal Z}_{\mathrm{CP}}
=
\mathcal Z_{\mathrm{CP}}.
\)
Combining the two relations gives
\(
\widehat{\mathcal Z}_{\mathrm{CP}}
=
\mathcal Z_{\mathrm{CP}}
\subseteq
\mathcal Z_{\mathrm P}.
\)

It remains to prove the covariance reduction. In the original unconstrained
PZ, the stochastic component is \(G_s\nu\), where
\(\nu\sim\mathcal N(0,I_{m_s})\). Hence,
\(
\Sigma_{\mathrm{orig}}=G_sG_s^\top .
\)
By Lemma~\ref{lem:generalized_constraint_absorption}, imposing the stochastic
part of the equality yields
\(
\nu=A_s^+(b-A_d\alpha)+N\zeta,
\,
\zeta\sim\mathcal N(0,I_{m_s-r}),
\)
where \(N\) is an orthonormal basis for \(\ker(A_s)\). Therefore, the
remaining stochastic component is \(G_sN\zeta\), whose covariance is
\(
\Sigma_c=G_sNN^\top G_s^\top .
\)
Since \(N\) has orthonormal columns, \(NN^\top\) is the orthogonal projector
onto \(\ker(A_s)\). Thus,
\(
P_s:=I_{m_s}-NN^\top
\)
is the orthogonal projector onto \(\operatorname{Im}(A_s^\top)\), and
\(
\Sigma_{\mathrm{orig}}-\Sigma_c
=
G_sG_s^\top-G_sNN^\top G_s^\top
=
G_sP_sG_s^\top\succeq0.
\)
Finally, since \(\operatorname{rank}(A_s)=r\), the projector onto
\(\operatorname{Im}(A_s^\top)\) has rank \(r\), so
\(
\operatorname{rank}(P_s)=r.
\)
\end{proof}

\begin{lemma}
\label{lem:exact_cpz_intersection}
Let
\(\mathcal Z^{i}_{\mathcal{CP}}=\zono{c_i,G_d^{i},G_s^{i},A_d^{i},A_s^{i}}_{\mathcal{CP}}\),
\(i=1,2\), be two CPZs as in Definition~\ref{def:cmixed_zonotope}, with
constraints \(A_d^{i}\alpha_i+A_s^{i}\nu_i=b_i\). Then their intersection is
\emph{exactly} a CPZ,
\[
\mathcal Z^{1}_{\mathrm{CP}}\cap\mathcal Z^{2}_{\mathrm{CP}}
=\zono{\hat c,\hat G_d,\hat G_s,\hat A_d,\hat A_s}_{\mathrm{CP}},
\]
where, with \(\hat\alpha=[\alpha_1^\top\ \alpha_2^\top]^\top\) and
\(\hat\nu=[\nu_1^\top\ \nu_2^\top]^\top\),
\begin{align}
&\hat c=c_1,\quad
\hat G_d=[\,G_d^{1}\ \ 0\,],\quad
\hat G_s=[\,G_s^{1}\ \ 0\,],\nonumber\\
&\hat A_d=
\begin{bmatrix}A_d^{1}&0\\ 0&A_d^{2}\\ G_d^{1}&-G_d^{2}\end{bmatrix},\
\hat A_s=
\begin{bmatrix}A_s^{1}&0\\ 0&A_s^{2}\\ G_s^{1}&-G_s^{2}\end{bmatrix},\
\hat b=
\begin{bmatrix}b_1\\ b_2\\ c_2-c_1\end{bmatrix}.\nonumber
\end{align}
The representation is exact: \(z\in\mathcal Z^{1}_{\mathrm{CP}}\cap\mathcal
Z^{2}_{\mathrm{CP}}\) if and only if \(z=\hat c+\hat G_d\hat\alpha+\hat
G_s\hat\nu\) for some \((\hat\alpha,\hat\nu)\) with
\(\hat A_d\hat\alpha+\hat A_s\hat\nu=\hat b\),
\(\hat\alpha\in[-1,1]^{m_d^{1}+m_d^{2}}\) and \(\hat\nu\sim\mathcal N(0,I)\).
\end{lemma}

% \begin{proof}
% This extends the constrained-zonotope intersection of~\cite{scott2016constrained}
% to the probabilistic case. A point \(z\) belongs to both sets if and only if
% there exist admissible pairs \((\alpha_1,\nu_1)\) and \((\alpha_2,\nu_2)\) such
% that
% \(z=c_1+G_d^{1}\alpha_1+G_s^{1}\nu_1=c_2+G_d^{2}\alpha_2+G_s^{2}\nu_2\). The first
% equality represents \(z\) through \((\hat G_d,\hat G_s)\); the two source
% constraints, together with the coupling
% \(G_d^{1}\alpha_1-G_d^{2}\alpha_2+G_s^{1}\nu_1-G_s^{2}\nu_2=c_2-c_1\), are exactly
% \(\hat A_d\hat\alpha+\hat A_s\hat\nu=\hat b\). No relaxation is introduced, hence
% the intersection is represented exactly.
% \end{proof}

\begin{proof}
This extends the constrained-zonotope intersection of~\cite{scott2016constrained}
to the probabilistic case. We prove the two inclusions. First, let
\(z\in\mathcal Z^{1}_{\mathrm{CP}}\cap\mathcal Z^{2}_{\mathrm{CP}}\).
Then, there exist admissible coefficients
\(\alpha_1\in[-1,1]^{m_d^1}\), \(\alpha_2\in[-1,1]^{m_d^2}\), and Gaussian
latent variables \(\nu_1,\nu_2\), satisfying the two original constraints
\(
A_d^1\alpha_1+A_s^1\nu_1=b_1,
\,
A_d^2\alpha_2+A_s^2\nu_2=b_2,
\)
such that
\[
z=c_1+G_d^1\alpha_1+G_s^1\nu_1
=
c_2+G_d^2\alpha_2+G_s^2\nu_2 .
\]
Define
\(
\hat\alpha=\operatorname{col}(\alpha_1,\alpha_2),
\,
\hat\nu=\operatorname{col}(\nu_1,\nu_2).
\)
Then, \(z\) can be represented using the first CPZ as
\[
z
=
c_1+
[\,G_d^1\;\;0\,]\hat\alpha
+
[\,G_s^1\;\;0\,]\hat\nu
=
\hat c+\hat G_d\hat\alpha+\hat G_s\hat\nu .
\]
Moreover, the two source constraints give the first two block rows of
\(\hat A_d\hat\alpha+\hat A_s\hat\nu=\hat b\). The equality of the two
representations gives
\(
G_d^1\alpha_1-G_d^2\alpha_2
+
G_s^1\nu_1-G_s^2\nu_2
=
c_2-c_1,
\)
which is exactly the third block row of
\(\hat A_d\hat\alpha+\hat A_s\hat\nu=\hat b\). Hence,
\(z\in\zono{\hat c,\hat G_d,\hat G_s,\hat A_d,\hat A_s,\hat b}_{\mathrm{CP}}\).

Conversely, let
\(
z
=
\hat c+\hat G_d\hat\alpha+\hat G_s\hat\nu
\)
for some
\(\hat\alpha=\operatorname{col}(\alpha_1,\alpha_2)\) and
\(\hat\nu=\operatorname{col}(\nu_1,\nu_2)\) satisfying
\(\hat A_d\hat\alpha+\hat A_s\hat\nu=\hat b\). By the first two block rows,
\(
A_d^1\alpha_1+A_s^1\nu_1=b_1,
\,
A_d^2\alpha_2+A_s^2\nu_2=b_2.
\)
Thus, the coefficient pairs are admissible for the two original CPZs. Also,
by the definition of \(\hat c,\hat G_d,\hat G_s\),
\(
z=c_1+G_d^1\alpha_1+G_s^1\nu_1 .
\)
The third block row gives
\(
G_d^1\alpha_1-G_d^2\alpha_2
+
G_s^1\nu_1-G_s^2\nu_2
=
c_2-c_1,
\)
and therefore
\(
c_1+G_d^1\alpha_1+G_s^1\nu_1
=
c_2+G_d^2\alpha_2+G_s^2\nu_2 .
\)
Hence, the same point \(z\) also belongs to
\(\mathcal Z^{2}_{\mathrm{CP}}\). Therefore,
\(z\in\mathcal Z^{1}_{\mathrm{CP}}\cap\mathcal Z^{2}_{\mathrm{CP}}\). The two inclusions prove equality. No relaxation or overapproximation is
introduced.
\end{proof}

\section{{Reachability Using Probabilistic Zonotopes}}

Consider the discrete-time system
\begin{align} \label{syst}
x_{k+1} = A^\star x_k+B^\star u_k + w_k,
\end{align}
where $A^\star$ and $B^\star$ are unknown systems' state-transition and input matrices, respectively, $x_k \in \mathbb{R}^n$, $u_k \in \mathbb{R}^m$, and $w_k$ is the disturbance. 

The following (possibly conservative) prior knowledge is assumed regarding the disturbance and the system parameters.

{\begin{assumption}
\label{PK}
Let $\theta^\star = [A^\star\ \ B^\star] \in \mathbb{R}^{n \times (n+m)}$. Then,
\(
% \label{eq:theta_mixed}
\mathrm{vec}(\theta^\star) \in \Theta
\), where
\begin{align}
% \label{ass:theta}
& \Theta
= 
\Big\{
\mathrm{vec}(\theta)=c_{\theta} +  G_{\theta,d}\alpha + G_{\theta,s}\nu : \nonumber \\ & \quad \quad \quad
\alpha_\theta\in[-1,1]^{m_{\theta_d}},
\,
\nu_\theta\sim\mathcal N(0,I_{m_{\theta_s}})
\Big\},
\end{align}
where
$c_{\theta} \in \mathbb{R}^{n(n+m)}$,
$G_{\theta,d} \in \mathbb{R}^{n(n+m)\times m_{\theta_d}}$,
$G_{\theta,s} \in \mathbb{R}^{n(n+m)\times m_{\theta_s}}$
are known prior parameters.
\end{assumption}}

\begin{assumption} \label{distass}
The disturbance $w_k$ belongs to the probabilistic zonotope $\mathcal W$ defined as
\begin{align} \label{Zdist}
& \mathcal W
= 
\Big\{
w=c+G_{d}\alpha+G_{s}\nu:\  \nonumber \\ & \quad \quad \quad
\alpha\in[-1,1]^{m_d},\ 
\nu\sim\mathcal N(0,I_{m_s})
\Big\},
\end{align}
where $w,c\in\mathbb R^{n_w}$,
$G_d\in\mathbb R^{n_w\times m_d}$,
$G_s\in\mathbb R^{n_w\times m_s}$,
$\alpha\in\mathbb R^{m_d}$,
and $\nu\in\mathbb R^{m_s}$.
\end{assumption}

For disturbances specifically, consider the decomposition $w_k = d_k + G_s\nu_k$,
where $d_k \in \mathcal{Z}(c_d, G_d)$ is a bounded deterministic term
and $\nu_k \sim \mathcal{N}(0, I_{m_s})$ is additive Gaussian noise.
Their sum is directly a PZ,
\begin{align}
  w_k &= d_k + G_s\nu_k
  \;\in\;
  \Bigl\{c_d + G_d\alpha + G_s\nu :
  \nonumber\\
  &\qquad \alpha \in [-1,1]^{m_d},\;
  \nu \sim \mathcal{N}(0, I_{m_s})\Bigr\}
  \;=\; \mathcal{Z}_{\mathrm{P}},
\end{align}
which is consistent with the structure of $\mathcal{W}$ in Assumption~\ref{distass}.

% Reachable sets are constructed by jointly leveraging two complementary components: 
% (i) a data-consistent set of system models, representing unknown but fixed realized  disturbances that are refined using probabilistic zonotopes to reduce epistemic 
% uncertainty, and (ii) an admissible disturbance set, representing  persistent aleatory uncertainty as per Assumption~\ref{distass}. While realized 
% uncertainty can be significantly reduced through data-driven refinement, aleatory uncertainty is inherently stochastic. However, we leverage realized sequences to probabilistically refine the conservative a priori aleatory uncertainty, driving the representation closer to the true underlying distribution.

Next, we extend standard approaches for reachable set computation to PZ disturbance sets and highlight their limitations in handling the two distinct sources of uncertainty considered in this paper.

Estimating the reachable set of a system with uncertain dynamics requires first characterizing the set of all system models consistent with the observed data. To this end, the following input-state trajectory of length \TC{$T$}  is collected:
\begin{align}
X_0 &= \begin{bmatrix} x_0 & x_1 & \cdots & x_{T-1} \end{bmatrix} \in \mathbb{R}^{n \times T}, \nonumber\\
X_1 &= \begin{bmatrix} x_1 & x_2 & \cdots & x_T \end{bmatrix} \in \mathbb{R}^{n \times T}, \nonumber\\
U_0 &= \begin{bmatrix} u_0 & u_1 & \cdots & u_{T-1} \end{bmatrix} \in \mathbb{R}^{m \times T}. \label{data}
\end{align}

\noindent Now, define the unmeasurable noise term by
\begin{align}\label{W}
W \;=\; \begin{bmatrix}w_0 & w_1 & \cdots & w_{T-1}\end{bmatrix}\in\mathbb{R}^{n\times T},
\end{align}

\begin{assumption}\label{PE}
The data matrix $X_0$ is full row rank (i.e., persistently exciting)
and $\mathrm{rank}(D_0)=n+m$, where
\begin{align} \label{D0}
    D_0 = \begin{bmatrix} X_0 \\[2pt] U_0 \end{bmatrix}\in\mathbb{R}^{(n+m)\times T}.
\end{align}
\end{assumption}

\noindent Assumption~\ref{PE} ensures that the recorded trajectory is sufficiently informative for the control design purpose. \vspace{6pt}

Lemma~\ref{lem:AB_mixed_short} provides a superset $\mathcal{S}_{AB}$ 
that contains all matrices $[A\;B]$ consistent with the measured 
trajectory and the prior disturbance band $\mathcal{W}$.

\begin{lemma}
\label{lem:AB_mixed_short}
Consider the system~\eqref{syst} and a measured trajectory~\eqref{data}
satisfying Assumption~\ref{PE}. Suppose each realized disturbance
\(w_k\), \(k=0,\ldots,T-1\), is described by an independent copy of the
same prior disturbance set \(\mathcal W\) in~\eqref{Zdist}. Let
\(G_d=[g_d^{(1)}\;\cdots\;g_d^{(m_d)}]\) and
\(G_s=[g_s^{(1)}\;\cdots\;g_s^{(m_s)}]\). Define
\(C_W:=c\mathbf 1_T^\top\). Let \(\gamma_{d,W}:=Tm_d\) and
\(\gamma_{s,W}:=Tm_s\). For \(i=1,\ldots,\gamma_{d,W}\), write
\(i=km_d+q\), with \(k=0,\ldots,T-1\) and \(q=1,\ldots,m_d\), and define
\begin{align}
G_{d,W}^{(i)}
:=
\bigl[
0_{n\times k}\;\;
g_d^{(q)}\;\;
0_{n\times(T-k-1)}
\bigr].
\label{eq:GdW_i_def}
\end{align}
Similarly, for \(j=1,\ldots,\gamma_{s,W}\), write
\(j=km_s+r\), with \(k=0,\ldots,T-1\) and \(r=1,\ldots,m_s\), and define
\begin{align}
G_{s,W}^{(j)}
:=
\bigl[
0_{n\times k}\;\;
g_s^{(r)}\;\;
0_{n\times(T-k-1)}
\bigr].
\label{eq:GsW_j_def}
\end{align}
Then, a superset of all data-consistent matrices \([A\;B]\) is
\begin{align}
\mathcal S_{AB}
=
&\Bigl\{
\theta
=
\Bigl(
X_1-C_W
-\sum_{i=1}^{\gamma_{d,W}}\alpha^{(i)}G_{d,W}^{(i)}
- 
\sum_{j=1}^{\gamma_{s,W}}\nu^{(j)}G_{s,W}^{(j)}
\Bigr) \nonumber\\ & D_0^\dagger:
\hspace{1.2em}
\alpha\in[-1,1]^{\gamma_{d,W}},\;
\nu\sim\mathcal N(0,I_{\gamma_{s,W}})
\Bigr\}.
\label{eq:AB_short}
\end{align}
\end{lemma}

\begin{proof}

Using the matrix generators in~\eqref{eq:GdW_i_def} and
\eqref{eq:GsW_j_def}, the stacked disturbance matrix \(W=[w_0\;\cdots\;w_{T-1}]\) belongs to the MPZ
\begin{align}
\mathcal Z_{\mathrm{MP}}
=
\Bigl\{
Y
&=
C_W
+\sum_{i=1}^{\gamma_{d,W}}\alpha^{(i)}G_{d,W}^{(i)}
+\sum_{j=1}^{\gamma_{s,W}}\nu^{(j)}G_{s,W}^{(j)}:
\nonumber\\[-1mm]
&\hspace{1.2em}
\alpha\in[-1,1]^{\gamma_{d,W}},\;
\nu\sim\mathcal N(0,I_{\gamma_{s,W}})
\Bigr\}.
\label{eq:W_MPZ}
\end{align}
The measured data \eqref{data} satisfy
\(X_1=[A\;B]D_0+W\). Since \(\operatorname{rank}(D_0)=n+m\) by
Assumption~\ref{PE}, \(D_0D_0^\dagger=I_{n+m}\), and therefore
\([A\;B]=(X_1-W)D_0^\dagger\). Substituting the MPZ parametrization
of \(W\) from~\eqref{eq:W_MPZ} into this affine reconstruction map gives
\eqref{eq:AB_short}.
\end{proof}

The superset obtained in Lemma~\ref{lem:AB_mixed_short}, is utilized by Algorithm~\ref{alg:hp_reach_single_window} to compute reachable set
over-approximations, with guarantees provided in
Theorem~\ref{thm:hp_reach_single_window}.

{ \begin{theorem}
\label{thm:hp_reach_single_window}
Consider the system~\eqref{syst} with trajectory data~\eqref{data}
satisfying Assumption~\ref{PE}. Let \(\mathcal{S}_{AB}(\delta)\) and
\(\mathcal{Z}_{\mathrm{MP}}(\delta)\), as high-confidence sets of \eqref{eq:AB_short} and \eqref{eq:W_MPZ}, respectively, constructed from Definition~\ref{def:confidence_set}. Define
\(
p_\delta:=1-\delta
\) and suppose $\mathbb{P}\bigl([A^\star\;B^\star]\in\mathcal{S}_{AB}(\delta)\bigr)
\ge p_\delta$, $\mathbb{P}\bigl(w_k\in\mathcal{Z}_{\mathrm{MP}}(\delta)\bigr)
\ge
p_\delta$, and
$k\ge 0$, which holds by Definition~\ref{def:confidence_set}.
Algorithm~\ref{alg:hp_reach_single_window} implements
\begin{equation}
\label{eq:hp_recursion_single}
\hat{\mathcal{R}}_{0}=\mathcal{X}_0,\ \text{and }
\hat{\mathcal{R}}_{k+1}
=
\mathcal{S}_{AB}(\delta)
\bigl(\hat{\mathcal{R}}_k\times\mathcal{U}_k\bigr)
\oplus \mathcal{Z}_{\mathrm{MP}}(\delta).
\end{equation}
and yields
\begin{equation}
\label{eq:hp_reach_single_window_product}
\mathbb{P}\!\left(
\mathcal{R}_t\subseteq \hat{\mathcal{R}}_t,\;
t=0,\ldots,k
\right)
\ge
p_\delta^{\,k+1}.
\end{equation}
\end{theorem}

\begin{proof}
Let \(E_{AB}:=\{[A^\star\;B^\star]\in\mathcal{S}_{AB}(\delta)\}\) and
\(E_{W,t}:=\{w_t\in\mathcal{Z}_{\mathrm{MP}}(\delta)\}\).  \(\{E_{w,t}(\delta)\}_{t\ge0}\) are mutually independent and
are independent of \(E_{AB}\), and, by
Definition~\ref{def:confidence_set},
\(\mathbb{P}(E_{AB})\ge p_\delta\) and \(\mathbb{P}(E_{W,t})\ge p_\delta\). For horizon \(k\),
define
\(
E_k
=
E_{AB}\cap\bigcap_{t=0}^{k-1}E_{W,t},
\)

with the empty intersection interpreted as the sure event. On \(E_k\), the containment follows by induction. The base case holds
because \(\hat{\mathcal R}_0=\mathcal X_0\). If
\(x_t\in\hat{\mathcal R}_t\) for some \(t<k\), then
\([A^\star\;B^\star]\in\mathcal S_{AB}(\delta)\) and
\(w_t\in\mathcal{Z}_{\mathrm{MP}}(\delta)\). Hence, for any
\(u_t\in\mathcal U_t\),
\begin{align}
x_{t+1}
&=
[A^\star\;B^\star]
\begin{bmatrix}x_t\\u_t\end{bmatrix}
+w_t
\notag\\
&\in
\mathcal{S}_{AB}(\delta)
\bigl(\hat{\mathcal R}_t\times\mathcal U_t\bigr)
\oplus \mathcal{Z}_{\mathrm{MP}}(\delta)
=
\hat{\mathcal R}_{t+1}.
\label{eq:one_step_reach_membership}
\end{align}
Thus, \(E_k\) implies
\(\mathcal R_t\subseteq\hat{\mathcal R}_t\) for all \(t=0,\ldots,k\), and
therefore
\begin{equation}
\label{eq:containment_from_prefix_event}
\mathbb{P}\!\left(
\mathcal{R}_t\subseteq \hat{\mathcal{R}}_t,\;
t=0,\ldots,k
\right)
\ge
\mathbb{P}(E_k).
\end{equation}
The independence of the events yields
\begin{align}
\mathbb{P}(E_k)
&=
\mathbb{P}(E_{AB})
\prod_{t=0}^{k-1}\mathbb{P}(E_{W,t})
\ge
p_\delta^{\,k+1}.
\label{eq:prefix_event_product_bound}
\end{align}
Combining~\eqref{eq:containment_from_prefix_event}
and~\eqref{eq:prefix_event_product_bound} gives
\eqref{eq:hp_reach_single_window_product}.
\end{proof}

\begin{algorithm}[t]
\caption{Conservative High-probability reachability analysis (single data window)}
\label{alg:hp_reach_single_window}
\begin{algorithmic}[1]
\Require Initial set $\mathcal{X}_0$, input sets $\{\mathcal{U}_k\}_{k=0}^{N-1}$,
confidence level $\delta\in(0,1)$,
data $(x_0,u_0),\ldots,(x_{T-1},u_{T-1})$ satisfying
Assumption~\ref{PE}
\Ensure Sets $\{\hat{\mathcal{R}}_k\}_{k=0}^{N}$
\State Construct the $\delta$-confidence model set
$\mathcal{S}_{AB}(\delta)$ from the data
via Definition~\ref{def:confidence_set} applied to
Lemma~\ref{lem:AB_mixed_short}.
\State Construct the $\delta$-confidence disturbance set for $\mathcal{Z}_{\mathrm{MP}}$ in \eqref{eq:W_MPZ} 
by applying Definition~\ref{def:confidence_set} and call it $\mathcal{Z}_{\mathrm{MP}}(\delta)$.
\State Set $\hat{\mathcal{R}}_0 \gets \mathcal{X}_0$.
\For{$k=0,1,\ldots,N-1$}
  \State $\hat{\mathcal{R}}_{k+1} \gets
  \mathcal{S}_{AB}(\delta)\bigl(\hat{\mathcal{R}}_k\times\mathcal{U}_k\bigr)
  \oplus \mathcal{Z}_{\mathrm{MP}}(\delta)$.
  \label{step:hp_recursion}
\EndFor
\State \Return $\{\hat{\mathcal{R}}_k\}_{k=0}^{N}$.
\end{algorithmic}
\end{algorithm}

Algorithm 1 is conservative because it does not exploit available information to refine the two distinct
sources of uncertainty that enter the reachable set. We now describe each source and outline how it can be reduced.

\begin{list4}
    \item \textbf{Realized disturbance sequence.}
    The disturbances \(W=[w_0,\ldots,w_{T-1}]\) that occurred during data
    collection are fixed after realization but remain unknown. The conservative
 approach in Algorithm 1 assigns every realized \(w_k\) the same full prior disturbance set \(\mathcal W\), which can be much larger than what is consistent with the measured trajectory. Instead, we estimate each realized disturbance using two sources of information: the model prior, which induces a sample-dependent set \(\mathcal W_{\theta,k}\), and the data equation, which enforces trajectory consistency across the collected window. This removes disturbance realizations that cannot explain the data and yields refined realized-disturbance proxies. These proxies are then used to construct a tighter data-consistent model set.

    \item \textbf{Aleatory disturbance set.}
   Future disturbances are different from past realized disturbances: they are not
fixed samples that can be reconstructed from data, but persistent stochastic
effects that continue to act during online operation. These future disturbances
and their true distribution are typically unknown; at the beginning, only a
conservative surrogate disturbance set is available. Finite data cannot
determine future disturbances exactly, but the refined realized-disturbance
proxies provide informative samples of the underlying disturbance mechanism.
Once enough such proxies are available, they are used to contract the
conservative prior disturbance set into a smaller high-confidence aleatory set
\(\mathcal W_\star(\delta)\), which is then used for future reachable-set
propagation.
\end{list4}

\noindent Together, refining both sources reduces the conservatism of
the reachable-set overapproximation from two complementary
directions: tighter realized disturbance sets yield a more accurate
model set $\mathcal{S}_{AB}(\delta)$, while a tighter aleatory set
directly shrinks the additive uncertainty propagated at each step.

% This distinction motivates two fundamental questions:
% \begin{itemize}
%     \item \textbf{Realized uncertainty reduction.}
%     How can prior knowledge and data be leveraged to estimate the realized
%     disturbance sequence $W$ more accurately, such that each element is confined
%     to a probabilistic zonotope with substantially smaller bounds and variance, yielding
%     an estimate of $W$ that is significantly tighter than the disturbance set
%     obtained by concatenating disturbances satisfying \eqref{Zdist}?

%     \item \textbf{Aleatory uncertainty reduction.}
%     How can these estimates be used to refine the conservative disturbance set
%     $\mathcal W$ defined in \eqref{Zdist} into a tighter admissible set
%     $\widehat{\mathcal W}$?
% \end{itemize}

% By constructing such refined estimates, the online disturbance $w_k$ can be
% assumed to satisfy $w_k \in \widehat{\mathcal W}$, which substantially reduces
% conservativeness in control design. In this manner, prior knowledge and collected
% data are leveraged not only to mitigate multiplicative uncertainty arising from
% the realized disturbance sequence, but also to tighten the admissible set of
% additive disturbances. {The next section shows how to leverage data and prior knowledge to reduce uncertainties and thus shrink the reachable set.}

\section{Uncertainty Reduction}
\subsection{Reducing the Realized Uncertainties}
% {This section introduces two lemmas that eliminate disturbance realizations inconsistent with the observed data (Lemma~\ref{lem:datacons_nzmean}) and with prior knowledge on the system parameters (Lemma~\ref{lem:prior_to_wk}). Each lemma yields a probabilistic zonotope representation for a given disturbance realization, shifting the center toward the true but unknown disturbance and reducing uncertainty by appropriately modifying the generators.}

{This section introduces two lemmas that eliminate disturbance realizations inconsistent with the prior knowledge on the system parameters (Lemma~\ref{lem:prior_to_wk}), and with the observed data (Lemma~\ref{lem:datacons_nzmean}). Their combination yields a CPZ representation for each disturbance realization, shifting its center toward the true but unknown value and reducing uncertainty by modifying the generators.

{
\begin{lemma}
\label{lem:prior_to_wk}
Consider the system~\eqref{syst}. Define
\begin{equation}
\label{eq:regressor_phi}
\phi_k :=
\begin{bmatrix} x_k \\ u_k \end{bmatrix}
\in \mathbb{R}^{n+m},
\qquad k=0,\ldots,T-1,
\end{equation}
and let
\(\theta^\star := [A^\star\ B^\star]\in\mathbb{R}^{n\times(n+m)}\).
Suppose Assumptions~\ref{PK} and~\ref{distass} hold. Then, for each fixed
\(k=0,\ldots,T-1\), the realized disturbance \(w_k\) consistent with the
model prior belongs to
\begin{align}
\mathcal W_{\theta,k}
:=
\Bigl\{
&c_k+G_{d,k}\alpha_{\theta}
+G_{s,k}\nu_{\theta}
\; : \;
\alpha_{\theta}\in[-1,1]^{m_{\theta_d}},
\nonumber\\
&\nu_{\theta}\sim\mathcal N(0,I_{m_{\theta_s}})
\Bigr\},
\label{eq:wtheta_k_set}
\end{align}
where
\begin{align}
M_k &:= \phi_k^\top\otimes I_n , \,\,
c_k := x_{k+1}-M_kc_\theta ,
\nonumber\\
G_{d,k} &:= -M_kG_{\theta,d},
\,\,
G_{s,k}:= -M_kG_{\theta,s}.
\label{eq:wk_params_from_theta}
\end{align}
\end{lemma}

\begin{proof}
For each fixed \(k\), the dynamics give
\(w_k=x_{k+1}-\theta^\star\phi_k\). Since
\(\theta^\star\phi_k=(\phi_k^\top\otimes I_n)\operatorname{vec}(\theta^\star)
=M_k\operatorname{vec}(\theta^\star)\), we have
\begin{align}
w_k=x_{k+1}-M_k\operatorname{vec}(\theta^\star).
\label{eq:wk_affine_in_vec_theta}
\end{align}
By Assumption~\ref{PK}, there exist common latent variables
\(\alpha_{\theta}\in[-1,1]^{m_{\theta_d}}\) and
\(\nu_{\theta}\sim\mathcal N(0,I_{m_{\theta_s}})\), independent of \(k\),
such that
\(\operatorname{vec}(\theta^\star)
=c_\theta+G_{\theta,d}\alpha_{\theta}
+G_{\theta,s}\nu_{\theta}\). Substitution into
\eqref{eq:wk_affine_in_vec_theta} gives
\begin{align}
w_k
&=
x_{k+1}
-
M_k
\bigl(
c_\theta+G_{\theta,d}\alpha_{\theta}
+G_{\theta,s}\nu_{\theta}
\bigr)
\nonumber\\
&=
c_k+G_{d,k}\alpha_{\theta}
+G_{s,k}\nu_{\theta},
\end{align}
with the parameters in~\eqref{eq:wk_params_from_theta}. Hence,
\(w_k\in\mathcal W_{\theta,k}\).
\end{proof}
}

We now further refine each realized disturbance estimate by enforcing consistency with the measured trajectory data.

{
\begin{lemma}
\label{lem:datacons_nzmean}
Consider the system~\eqref{syst} and a measured trajectory satisfying
Assumption~\ref{PE}, with \(T>n+m\). At update window \(j\), let
\(\mathcal W^{(j)}=\langle c_j,G_{d,j},G_{s,j}\rangle_{\mathcal P}\)
denote the frozen admissible disturbance prior, and let each \(w_k\),
\(k=0,\ldots,T-1\), be described by a copy of this set. Write
\begin{align}
G_{d,j}
&=
\bigl[
g_{d,j}^{(1)}\;\cdots\;g_{d,j}^{(\bar m_d)}
\bigr],
\qquad
G_{s,j}
=
\bigl[
g_{s,j}^{(1)}\;\cdots\;g_{s,j}^{(\bar m_s)}
\bigr].
\label{eq:local_generator_columns}
\end{align}
Define the stacked center
\begin{align}
C_N
:=
[c_j\;\cdots\;c_{j}]
\in\mathbb R^{n\times T}.
\label{eq:CN_matrix_center}
\end{align}
Let \(\gamma_{d,N}:=T\bar m_d\) and
\(\gamma_{s,N}:=T\bar m_s\). For each
\(i=1,\ldots,\gamma_{d,N}\), write
\(i=k\bar m_d+q\), with \(k=0,\ldots,T-1\) and
\(q=1,\ldots,\bar m_d\), and define
\begin{align}
G_{d,N}^{(i)}
:=
\bigl[
0_{n\times k}\;\;
g_{d,j}^{(q)}\;\;
0_{n\times(T-k-1)}
\bigr]
\in\mathbb R^{n\times T}.
\label{eq:GdN_i_explicit}
\end{align}
Similarly, for each \(\ell=1,\ldots,\gamma_{s,N}\), write
\(\ell=k\bar m_s+r\), with \(k=0,\ldots,T-1\) and
\(r=1,\ldots,\bar m_s\), and define
\begin{align}
G_{s,N}^{(\ell)}
:=
\bigl[
0_{n\times k}\;\;
g_{s,j}^{(r)}\;\;
0_{n\times(T-k-1)}
\bigr]
\in\mathbb R^{n\times T}.
\label{eq:GsN_j_explicit}
\end{align}
Collect the matrix generators as
\begin{align}
\widetilde G_{d,N}
&:=
\bigl[
G_{d,N}^{(1)}
\;\cdots\;
G_{d,N}^{(\gamma_{d,N})}
\bigr],
\nonumber\\
\widetilde G_{s,N}
&:=
\bigl[
G_{s,N}^{(1)}
\;\cdots\;
G_{s,N}^{(\gamma_{s,N})}
\bigr].
\label{eq:stacked_matrix_generators_compact}
\end{align}
Then, the stacked realized disturbance matrix satisfies
\(W\in\mathcal W_N^0\), where
\begin{align}
\mathcal W_N^0
:=
\Bigl\{
&Y
=
C_N
+
\sum_{i=1}^{\gamma_{d,N}}
\alpha^{(i)}G_{d,N}^{(i)}
+
\sum_{\ell=1}^{\gamma_{s,N}}
\nu^{(\ell)}G_{s,N}^{(\ell)}:
\nonumber\\
&\alpha\in[-1,1]^{\gamma_{d,N}},
\quad
\nu\sim\mathcal N(0,I_{\gamma_{s,N}})
\Bigr\}.
\label{eq:stacked_matrix_pz_prior}
\end{align}
Moreover, imposing data consistency yields CMPZ
\begin{align}
W\in\widehat{\mathcal W}_N^c
:=
\Bigl\{
Y\in\mathcal W_N^0\; ;\;
&\sum_{i=1}^{\gamma_{d,N}}\alpha^{(i)}A_{d,N}^{(i)}
+
\sum_{\ell=1}^{\gamma_{s,N}}\nu^{(\ell)}A_{s,N}^{(\ell)}
\nonumber\\[-1mm]
&=B_N
\Bigr\},
\label{eq:stacked_matrix_cmpz}
\end{align}
where
\begin{align}
A_{d,N}^{(i)}
&:=
G_{d,N}^{(i)}D_0^\perp,
\qquad
i=1,\ldots,\gamma_{d,N},
\nonumber\\
A_{s,N}^{(\ell)}
&:=
G_{s,N}^{(\ell)}D_0^\perp,
\qquad
\ell=1,\ldots,\gamma_{s,N},
\nonumber\\
B_N
&:=
(X_1-C_N)D_0^\perp .
\label{eq:cmpz_constraint_matrices}
\end{align}
Equivalently,
\[
\widehat{\mathcal W}_N^c
=
\langle
C_N,\widetilde G_{d,N},\widetilde G_{s,N},
\widetilde A_{d,N},\widetilde A_{s,N},B_N
\rangle_{\mathrm{CMP}},
\]
with
\(\widetilde A_{d,N}
=
[A_{d,N}^{(1)}\;\cdots\;A_{d,N}^{(\gamma_{d,N})}]\)
and
\(\widetilde A_{s,N}
=
[A_{s,N}^{(1)}\;\cdots\;A_{s,N}^{(\gamma_{s,N})}]\).
The set \(\widehat{\mathcal W}_N^c\) is nonempty on the confidence event
that the true realized disturbance sequence belongs to \(\mathcal W_N^0\).
\end{lemma}

\begin{proof}
For each \(k=0,\ldots,T-1\), the frozen prior gives
\begin{align}
w_k\in\mathcal W^{(j)}
=
\Bigl\{
&\bar w_k
=
c_j+G_{d,j}\alpha_k+G_{s,j}\nu_k:
\nonumber\\
&\alpha_k\in[-1,1]^{\bar m_d},
\quad
\nu_k\sim\mathcal N(0,I_{\bar m_s})
\Bigr\}.
\label{eq:single_proxy_for_matrix_stack}
\end{align}
Using the column decompositions in~\eqref{eq:local_generator_columns}, this
same set can be written as
\begin{align}
\mathcal W^{(j)}
=
\Bigl\{
&\bar w_k
=
c_j
+
\sum_{q=1}^{\bar m_d}
\alpha_k^{(q)}g_{d,j}^{(q)}
+
\sum_{r=1}^{\bar m_s}
\nu_k^{(r)}g_{s,j}^{(r)}:
\nonumber\\
&\alpha_k\in[-1,1]^{\bar m_d},
\quad
\nu_k\sim\mathcal N(0,I_{\bar m_s})
\Bigr\}.
\label{eq:wk_column_expansion_set}
\end{align}
Placing each element \(\bar w_k\) in the \((k+1)\)-st column and zeros in
all other columns gives the matrix generators
\(G_{d,N}^{(i)}\) and \(G_{s,N}^{(\ell)}\) defined in
\eqref{eq:GdN_i_explicit}--\eqref{eq:GsN_j_explicit}. Hence all stacked
disturbance matrices generated by the \(T\) copies of the frozen prior are
exactly represented by \(\mathcal W_N^0\) in~\eqref{eq:stacked_matrix_pz_prior}.
Therefore, for the realized disturbance matrix
\(W=[w_0\;\cdots\;w_{T-1}]\), one has \(W\in\mathcal W_N^0\).

It remains to impose data consistency. A candidate disturbance matrix \(Y\)
is data-consistent if, after subtracting it from the measured next-state
matrix \(X_1\), the remaining matrix can be explained by one constant model
\([A\;B]\) over all collected samples. Thus, there must exist
\([A\;B]\) such that \(X_1-Y=[A\;B]D_0\). Since
\(D_0D_0^\perp=0\), every data-consistent \(Y\) must satisfy
\begin{align}
(X_1-Y)D_0^\perp=0 .
\label{eq:data_consistency_matrix_new}
\end{align}
Substituting the representation of \(Y\in\mathcal W_N^0\) into
\eqref{eq:data_consistency_matrix_new} gives
\begin{align}
\sum_{i=1}^{\gamma_{d,N}}
\alpha^{(i)}G_{d,N}^{(i)}D_0^\perp
+
\sum_{\ell=1}^{\gamma_{s,N}}
\nu^{(\ell)}G_{s,N}^{(\ell)}D_0^\perp
=
(X_1-C_N)D_0^\perp .
\label{eq:matrix_affine_constraint_new}
\end{align}
Using the definitions in~\eqref{eq:cmpz_constraint_matrices}, this becomes
\begin{align}
\sum_{i=1}^{\gamma_{d,N}}
\alpha^{(i)}A_{d,N}^{(i)}
+
\sum_{\ell=1}^{\gamma_{s,N}}
\nu^{(\ell)}A_{s,N}^{(\ell)}
=
B_N .
\label{eq:cmpz_affine_constraint}
\end{align}
Therefore, imposing data consistency on \(\mathcal W_N^0\) gives the
constrained matrix probabilistic zonotope
\(\widehat{\mathcal W}_N^c\) in~\eqref{eq:stacked_matrix_cmpz}.

Finally, let \(W^\star\) be the true realized disturbance matrix. Since the
measured data were generated by the true system, one has
\(X_1=[A^\star\;B^\star]D_0+W^\star\). Right multiplication by
\(D_0^\perp\) gives \((X_1-W^\star)D_0^\perp=0\). On the confidence event,
\(W^\star\in\mathcal W_N^0\). Therefore, there exist admissible
\(\alpha^\star\in[-1,1]^{\gamma_{d,N}}\) and
\(\nu^\star\sim\mathcal N(0,I_{\gamma_{s,N}})\) satisfying
\eqref{eq:cmpz_affine_constraint}. Hence,
\(\widehat{\mathcal W}_N^c\) is nonempty.
\end{proof}
}

Lemma~\ref{lem:datacons_nzmean} provides a CMPZ representation for the full
stacked realized-disturbance matrix using the frozen disturbance prior and the trajectory-consistency constraint. This representation allows different realized disturbances through
column-wise latent variables, while the trajectory-consistency constraint couples these variables across the data window.
Using Remark~\ref{CMPZ-CPZ}, the CMPZ is first viewed as a vectorized CPZ. Lemma~\ref{lem:generalized_constraint_absorption} is then applied to this full vectorized CPZ, not independently to each column. This absorbs the
stochastic part of the global equality constraint and yields an equivalent representation whose remaining equality constraint depends only on
deterministic latent variables. This remaining equality may still couple the variables associated with different time samples. The \(k\)-th realized-disturbance proxy is obtained by projecting this
globally constrained representation onto the \(k\)-th column. Hence, the
projected proxy is sample-wise in its output variable, but it may still carry
the effect of the global deterministic equality through an existential constraint on the stacked deterministic variables. 

Independently, Lemma~\ref{lem:prior_to_wk} uses the current model prior to
construct model-consistent proxies
\(
\{\widehat{\mathcal W}_{\theta,k}^{c}\}_{k=0}^{T-1}.
\)
The final realized-disturbance proxy for each sample is obtained by fusing
the data-consistent projected proxy with the model-consistent proxy via
Lemma~\ref{lem:exact_cpz_intersection}. The remaining deterministic equality
constraint does not change the fusion argument; it is carried into the fused
CPZ as an admissibility constraint. Thus, the fusion 
\begin{align}
\widehat{\mathcal W}_{k}^{c}(\delta)
\supseteq
\widehat{\mathcal W}_{D,k}^{c}(\delta)
\cap
\widehat{\mathcal W}_{\theta,k}^{c}(\delta),
\qquad k=0,\ldots,T-1 .
\label{fuzed1}
\end{align}
forms a CPZ for reach realization $k$, but with an equality constraint depending on all deterministic latent variables. Then, the
CMPZ obtained from stacking these realized disturbances after fusing admits the
CPZ representation
\begin{align}
\widehat{\mathcal W}^{c}
=
\Bigl\{
&Y
=
C_N
+
\sum_{i=1}^{\gamma_{d,N}}\alpha^{(i)}G_{d,N}^{(i)}
+
\sum_{j=1}^{\gamma_{s,N}}\nu^{(j)}G_{s,N}^{(j)}:
\nonumber\\[-1mm] \bar A_{d,N}\alpha=\bar b_N, \,\,
&\alpha\in[-1,1]^{\gamma_{d,N}},\;
\nu\sim\mathcal N(0,I_{\gamma_{s,N}})
\Bigr\}.
\label{eq:W_CMPZ_derived}
\end{align}

}

\begin{lemma}
\label{lem:AB_from_CMPZ}
Consider the system~\eqref{syst} and a measured trajectory~\eqref{data}
satisfying Assumption~\ref{PE}. For any disturbance-matrix set
\(\mathcal W\subseteq\mathbb R^{n\times T}\), define the induced model set
\begin{align}
\mathcal S_{AB}(\mathcal W)
:=
\{(X_1-Y)D_0^\dagger:Y\in\mathcal W\},
\label{eq:SAB_general}
\end{align}
Let \eqref{eq:W_CMPZ_derived} be the stacked refined realized disturbances. 
Then, the refined data-consistent model set is
\begin{align}
\widehat{\mathcal S}_{AB}^{c}
=
\mathcal S_{AB}(\widehat{\mathcal W}^{c}).
\label{eq:SAB_refined_set_def}
\end{align}
Equivalently,
\begin{align}
\widehat{\mathcal S}_{AB}^{c}
=
\Bigl\{
\theta \; ;\;
&\theta
=
\Bigl(
X_1-C_N
-
\sum_{i=1}^{\gamma_{d,N}}\alpha^{(i)}G_{d,N}^{(i)}
\nonumber\\[-1mm]
&\hspace{-1.3em}
-
\sum_{j=1}^{\gamma_{s,N}}\nu^{(j)}G_{s,N}^{(j)}
\Bigr)D_0^\dagger:
\nonumber\\[-1mm] \bar A_{d,N}\alpha=\bar b_N, \,
&\alpha\in[-1,1]^{\gamma_{d,N}}, \,
\nu\sim\mathcal N(0,I_{\gamma_{s,N}})
\Bigr\}.
\label{eq:SAB_CMPZ}
\end{align}
Moreover, every \(\theta\in\widehat{\mathcal S}_{AB}^{c}\) is consistent
with the measured trajectory and the refined stacked disturbance set.
\end{lemma}

\begin{proof}
By Lemma~\ref{lem:datacons_nzmean}, the realized disturbance sequence is
represented as a stacked disturbance matrix
\(Y=[y_0\;\cdots\;y_{T-1}]\) belonging to the data-consistent CMPZ
\(\widehat{\mathcal W}_N^c\). After applying
Lemma~\ref{lem:generalized_constraint_absorption}, the stochastic part of the
equality constraint is absorbed, while the remaining deterministic equality
\(\bar A_{d,N}\alpha=\bar b_N\) is retained in
\(\widehat{\mathcal W}^{c}\). Using the data \eqref{data} and the system dynamics, any candidate
disturbance matrix \(Y\in\widehat{\mathcal W}^{c}\) induces
\begin{align}
\theta=[A\;B]=(X_1-Y)D_0^\dagger .
\label{eq:proof_theta_from_Y}
\end{align}
Thus, each \(Y\in\widehat{\mathcal W}^{c}\), with its admissible
\(\alpha\) satisfying \(\bar A_{d,N}\alpha=\bar b_N\), induces a model
matrix \(\theta=[A\;B]\) consistent with the measured data. Therefore,
\begin{align}
\widehat{\mathcal S}_{AB}^{c}
=
\{(X_1-Y)D_0^\dagger:Y\in\widehat{\mathcal W}^{c}\}
=
\mathcal S_{AB}(\widehat{\mathcal W}^{c}).
\label{eq:proof_SAB_affine_image}
\end{align}
Substituting the parametrization \eqref{eq:W_CMPZ_derived} into
\eqref{eq:proof_SAB_affine_image} gives \eqref{eq:SAB_CMPZ}. 
This completes the proof.
\end{proof}

\subsection{Refining the Disturbance Set}
{This subsection shows how to refine the disturbance set defined in Assumption~\ref{distass} given available estimates of realized disturbances. The central idea is to contract the disturbance set so that it contains all disturbances consistent with the observed realizations, while remaining a conservative subset of the original prior set.}

\begin{lemma}
\label{lem:exact_containment_set_refinement}
Let the conservative disturbance prior in Assumption~\ref{distass} be
\(
\mathcal W=\langle c,G_d,G_s\rangle_{\mathrm{P}},
\)
and let the refined realized-disturbance proxy sets be
\begin{align}
\widehat{\mathcal W}_j
=
\langle \hat c_j,\hat G_{d,j},\hat G_{s,j}\rangle_{\mathrm{P}},
\qquad j=0,\ldots,T-1 .
\end{align}
Fix a confidence level \(\delta\), and define, with the truncation widths
\(\gamma_\delta\) of \eqref{eq:trunc_width} (each computed from the respective
number of stochastic generators),
\begin{align}
G_\delta=[G_d\;\;\gamma_\delta G_s],
\qquad
\hat G_{j,\delta}=[\hat G_{d,j}\;\;\gamma_\delta\hat G_{s,j}].
\end{align}
Let
\begin{align}
\mathcal W(\delta)
=
\{c+G_\delta\xi:\|\xi\|_\infty\le1\}
=
\{w:H_\delta w\le h_\delta\}
\label{eq:prior_Hrep}
\end{align}
be an exact half-space representation of the prior confidence set, and let
\begin{align}
\widehat{\mathcal W}_j(\delta)
=
\{\hat c_j+\hat G_{j,\delta}\zeta_j:
\|\zeta_j\|_\infty\le1\}.
\end{align}
Consider the linear program
\begin{subequations}
\label{eq:scaled_Hrep_refinement_lp}
\begin{align}
\mathop{\operatorname{minimize}}_{y,\rho}\quad & \rho
\label{eq:scaled_Hrep_refinement_lp_a}\\
\mathrm{s.t.}\quad
&0\le \rho\le1,
\label{eq:scaled_Hrep_refinement_lp_b}\\
&H_\delta y\le (1-\rho)h_\delta,
\label{eq:scaled_Hrep_refinement_lp_c}\\
&H_\delta \hat c_j
+
|H_\delta \hat G_{j,\delta}|\mathbf 1
\le
H_\delta y+\rho h_\delta,
\nonumber\\[-1mm]
&\hspace{34mm}
j=0,\ldots,T-1 .
\label{eq:scaled_Hrep_refinement_lp_d}
\end{align}
\end{subequations}
Let \((y^\star,\rho^\star)\) be an optimizer and define
\begin{align}
\mathcal W_\star(\delta)
:=
y^\star\oplus \rho^\star\mathcal W(\delta)
=
\{y^\star+\rho^\star c+\rho^\star G_\delta\xi:
\|\xi\|_\infty\le1\}.
\label{eq:learned_scaled_prior_set}
\end{align}
Then,
\begin{align}
\mathcal W_\star(\delta)\subseteq \mathcal W(\delta),
\,
\widehat{\mathcal W}_j(\delta)\subseteq \mathcal W_\star(\delta),
\quad j=0,\ldots,T-1 .
\end{align}
\end{lemma}

\begin{proof}
By Definition~\ref{def:confidence_set}, the high-confidence disturbance
prior is the deterministic zonotope \(\mathcal W(\delta)\), which admits the
exact half-space representation given in~\eqref{eq:prior_Hrep}.
The learned disturbance set is defined as
\eqref{eq:learned_scaled_prior_set}.
Thus, \(\mathcal W_\star(\delta)\) uses exactly the same generator
directions as the prior set; only its translation \(y^\star\) and scale
\(\rho^\star\) are learned.

We first prove that
\(\mathcal W_\star(\delta)\subseteq\mathcal W(\delta)\). From
constraint~\eqref{eq:scaled_Hrep_refinement_lp_c}, we have \vspace{-12pt}
\begin{align}
H_\delta y^\star\le (1-\rho^\star)h_\delta .
\label{eq:proof_y_scaled_hrep}
\end{align}
Since \(0\le\rho^\star\le1\), this means that
\(y^\star\in(1-\rho^\star)\mathcal W(\delta)\). Equivalently, there exists
\(p^\star\in\mathcal W(\delta)\) such that
\begin{align}
y^\star=(1-\rho^\star)p^\star .
\label{eq:proof_anchor_relation}
\end{align}
Thus,
\(
p^\star:=\frac{y^\star}{1-\rho^\star}
\)
satisfies \(H_\delta p^\star\le h_\delta\) by \eqref{eq:prior_Hrep}. Therefore, for
\(0\le\rho^\star<1\), the learned set can be written as
\(
\mathcal W_\star(\delta)
=
(1-\rho^\star)p^\star
\oplus
\rho^\star\mathcal W(\delta),
\)
which is contained in \(\mathcal W(\delta)\) by convexity.

The remaining case is \(\rho^\star=1\). Then,
\eqref{eq:proof_y_scaled_hrep} gives \(H_\delta y^\star\le0\). For any
\(w\in\mathcal W_\star(\delta)\), there exists
\(z\in\mathcal W(\delta)\) such that \(w=y^\star+z\). Hence,
\(
H_\delta w
=
H_\delta y^\star+H_\delta z
\le
0+h_\delta
=
h_\delta .
\)
Thus, \(w\in\mathcal W(\delta)\), and consequently
\(\mathcal W_\star(\delta)\subseteq\mathcal W(\delta)\).

Now take any \(w\in\mathcal W_\star(\delta)\). By
\eqref{eq:learned_scaled_prior_set}, there exists
\(\xi\), with \(\|\xi\|_\infty\le1\), such that
\(
w
=
y^\star+\rho^\star(c+G_\delta\xi).
\)
Using \eqref{eq:proof_anchor_relation}, this becomes
\(
w
=
(1-\rho^\star)p^\star
+
\rho^\star(c+G_\delta\xi).
\)
Both \(p^\star\) and \(c+G_\delta\xi\) belong to
\(\mathcal W(\delta)\). Since \(\mathcal W(\delta)\) is convex and
\(0\le\rho^\star\le1\), the point \(w\) is a convex combination of two
points in \(\mathcal W(\delta)\). Hence \(w\in\mathcal W(\delta)\). Since
\(w\) was arbitrary, we obtain \(
\mathcal W_\star(\delta)\subseteq\mathcal W(\delta).
\)

We now prove that each refined proxy confidence set is contained in the
learned set. Fix \(j\in\{0,\ldots,T-1\}\). The \(j\)-th proxy confidence set
is
\(
\widehat{\mathcal W}_j(\delta)
=
\{\hat c_j+\hat G_{j,\delta}\zeta_j:
\|\zeta_j\|_\infty\le1\}.
\)
We first show that the learned set can be written in half-space form as
\begin{align}
\mathcal W_\star(\delta)
=
\{w:H_\delta w\le H_\delta y^\star+\rho^\star h_\delta\}.
\label{eq:proof_learned_hrep}
\end{align}
Note that
\(w\in y^\star\oplus\rho^\star\mathcal W(\delta)\) if and only if
\(w-y^\star\in\rho^\star\mathcal W(\delta)\). Since
\(\mathcal W(\delta)=\{z:H_\delta z\le h_\delta\}\), this is equivalent to
\(
H_\delta(w-y^\star)\le\rho^\star h_\delta,
\)
or equivalently
\(
H_\delta w\le H_\delta y^\star+\rho^\star h_\delta .
\) Therefore, to prove
\(\widehat{\mathcal W}_j(\delta)\subseteq\mathcal W_\star(\delta)\), it is
enough to show that every
\(w\in\widehat{\mathcal W}_j(\delta)\) satisfies the inequalities in
\eqref{eq:proof_learned_hrep}. Take an arbitrary
\(w\in\widehat{\mathcal W}_j(\delta)\). Then, there exists
\(\zeta_j\), with \(\|\zeta_j\|_\infty\le1\), such that
\(
w=\hat c_j+\hat G_{j,\delta}\zeta_j .
\label{eq:proof_proxy_element}
\)
Multiplying by \(H_\delta\) gives
\begin{align}
H_\delta w
=
H_\delta\hat c_j
+
H_\delta\hat G_{j,\delta}\zeta_j .
\label{eq:proof_proxy_hmap}
\end{align}
For each row \(i\), using \(\|\zeta_j\|_\infty\le1\), we have
\begin{align}
\bigl(H_\delta\hat G_{j,\delta}\zeta_j\bigr)_i
&\le
\sum_{\ell}
\left|
\bigl(H_\delta\hat G_{j,\delta}\bigr)_{i\ell}
\right|
\left|
(\zeta_j)_\ell
\right| \nonumber \\
&\le
\sum_{\ell}
\left|
\bigl(H_\delta\hat G_{j,\delta}\bigr)_{i\ell}
\right|.
\label{eq:proof_row_bound}
\end{align}
In vector form, this gives
\begin{align}
H_\delta\hat G_{j,\delta}\zeta_j
\le
|H_\delta\hat G_{j,\delta}|\mathbf 1 .
\label{eq:proof_vector_bound}
\end{align}
Combining \eqref{eq:proof_proxy_hmap} and
\eqref{eq:proof_vector_bound} yields
\(
H_\delta w
\le
H_\delta\hat c_j
+
|H_\delta\hat G_{j,\delta}|\mathbf 1 .
\)
By constraint~\eqref{eq:scaled_Hrep_refinement_lp_d},
\(
H_\delta\hat c_j
+
|H_\delta\hat G_{j,\delta}|\mathbf 1
\le
H_\delta y^\star+\rho^\star h_\delta .
\)
Therefore,
\(
H_\delta w
\le
H_\delta y^\star+\rho^\star h_\delta .
\)
By \eqref{eq:proof_learned_hrep}, this implies
\(w\in\mathcal W_\star(\delta)\). Since \(w\) was arbitrary, one has
\(
\widehat{\mathcal W}_j(\delta)
\subseteq
\mathcal W_\star(\delta).
\)
Since \(j\) was arbitrary, the containment holds for all
\(j=0,\ldots,T-1\).

Finally, the optimization minimizes \(\rho\). Thus, among all translated
and scaled copies of the prior confidence set that contain all refined proxy
confidence sets and remain inside the conservative prior, the optimizer
returns the tightest one according to the scale factor \(\rho\). This
completes the proof.
\end{proof}}

\begin{remark}
If the refined proxy \(\widehat{\mathcal W}_j\) is obtained from a projected
column of the stacked CMPZ, it may inherit a deterministic equality constraint
that involves latent variables from the full stacked disturbance sequence.
Thus, after stochastic constraint absorption, the proxy may have the form
\(\widehat{\mathcal W}_j=\{\hat c_j+\hat G_{d,j}\alpha+\hat G_{s,j}\nu \; ; \;
\hat A\alpha=\hat b,\;\alpha\in[-1,1]^{m_d},\;
\nu\sim\mathcal N(0,I_{m_s})\}\), where \(\alpha\) may contain coefficients
associated with all samples, not only the \(j\)-th one. Hence, the equality
constraint captures coupling among realized disturbances. Lemma~\ref{lem:exact_containment_set_refinement} can still be used after
replacing this equality-constrained proxy by a standard PZ
overapproximation. Let \(\alpha_0=\hat A^\dagger\hat b\), and let \(N_d\)
span \(\ker(\hat A)\). Then every feasible \(\alpha\) can be written as
\(\alpha=\alpha_0+N_d\beta\). The box constraint
\(\alpha\in[-1,1]^{m_d}\) may impose additional restrictions on \(\beta\);
dropping these restrictions and using a bounded box
\(\beta\in[-1,1]^{r_d}\), after scaling \(N_d\) if needed, gives a
conservative standard PZ
\(\overline{\mathcal W}_j
=\langle \bar c_j,\bar G_{d,j},\hat G_{s,j}\rangle_{\mathcal P}\), where
\(\bar c_j=\hat c_j+\hat G_{d,j}\alpha_0\) and
\(\bar G_{d,j}=\hat G_{d,j}N_d\). This removes the explicit deterministic
equality while preserving the outer containment
\(\widehat{\mathcal W}_j(\delta)\subseteq
\overline{\mathcal W}_j(\delta)\). Therefore, applying the LP in~\eqref{eq:scaled_Hrep_refinement_lp} to
\(\overline{\mathcal W}_j(\delta)\) guarantees containment of the original
coupled equality-constrained proxy as well. The learned set may be more
conservative, but
\(\widehat{\mathcal W}_j(\delta)\subseteq
\overline{\mathcal W}_j(\delta)\subseteq
\mathcal W_\star(\delta)\) remains valid.
\end{remark}

Because $T$ is finite, there remains a risk that a future disturbance realization falls outside the learned set $\mathcal{W}_\star^{(\gamma)}$. We bound this risk strictly using the scenario optimization framework.

\begin{lemma}
\label{lem:scenario_optimization_bound}
Assume the \(T\) proxy disturbance sets
\(\widehat{\mathcal{W}}_j\), \(j=0,\ldots,T-1\), are independent random
scenarios drawn from an underlying unknown distribution of disturbances.
Let \(\epsilon\in(0,1)\) denote the acceptable probability that a newly
drawn proxy set is not contained in the learned disturbance set, and let
\(\rho\in(0,1)\) denote the confidence parameter of the scenario guarantee.
Let the support rank be bounded by \(d=n_w+1\). If the number of collected
proxy sets \(T\) satisfies
\begin{equation}
    \sum_{i=0}^{d-1}
    {T \choose i}
    \epsilon^i(1-\epsilon)^{T-i}
    \le \rho,
    \label{eq:scenario_complexity}
\end{equation}
then, with probability at least \(1-\rho\) over the extraction of the
\(T\) proxy sets, the learned high-confidence disturbance set
\(\mathcal{W}_\star(\delta)\) covers a newly drawn unseen proxy set
\(\widehat{\mathcal{W}}_{\mathrm{new}}\) with probability at least
\(1-\epsilon\). Equivalently,
\begin{equation}
    \mathbb{P}^{T}\!\left(
    \mathbb{P}\!\left\{
    \widehat{\mathcal{W}}_{\mathrm{new}}
    \not\subseteq
    \mathcal{W}_\star(\delta)
    \right\}
    \le
    \epsilon
    \right)
    \ge
    1-\rho .
    \label{eq:scenario_generalization_bound}
\end{equation}
Here, \(\mathbb{P}^{T}\) is the product probability measure over the
\(T\) observed proxy sets. 
\end{lemma}

\begin{proof}
The proof follows from the scenario-based optimization framework
in~\cite{1632303}.
\end{proof}

{
\begin{remark}
\label{rem5}
Lemmas~\ref{lem:prior_to_wk}, \ref{lem:datacons_nzmean},
\ref{lem:AB_from_CMPZ}, and
\ref{lem:exact_containment_set_refinement} are used sequentially and on
different time scales in Algorithm~\ref{alg:refined_reach}. At update window
\(j\), the current aleatory disturbance prior \(\mathcal W^{(j)}\) is first
frozen. Lemma~\ref{lem:datacons_nzmean} uses this frozen prior and the
trajectory-consistency constraint to construct data-consistent
realized-disturbance proxies. In parallel, Lemma~\ref{lem:prior_to_wk} uses
the current model prior to construct model-consistent realized-disturbance
proxies. These two descriptions are fused sample-wise using
Lemma~\ref{lem:exact_cpz_intersection}.  Thus, the fused
proxy contains realized disturbances consistent with both the data-consistency path
and the model-prior path.

The fused proxies are then used in a batch-wise refinement loop. For the
current data batch, the proxies in~\eqref{fuzed1} are stacked and used in
Lemma~\ref{lem:AB_from_CMPZ} to identify a refined admissible system model
set. Once enough fused proxies have been collected,
Lemma~\ref{lem:exact_containment_set_refinement} also uses them to learn a
nested aleatory disturbance set
\(
\mathcal W_\star^{(j+1)}(\delta)
\subseteq
\mathcal W^{(j)}(\delta).
\)
At the next update window, the CMPZ learned model set is fused with the MPZ
model prior in Assumption~\ref{PK}. To this end,
Remark~\ref{CMPZ-CPZ} is first used to express the CMPZ and MPZ as a CPZ and
a PZ, respectively. Lemma~\ref{lem:generalized_constraint_absorption} is then
applied to the resulting CPZ to absorb the stochastic part of its equality
constraint, after which Lemma~\ref{lem:exact_cpz_intersection} is used to fuse the
two model descriptions. The fused set is used as the updated model prior in
Lemma~\ref{lem:prior_to_wk}.
This update introduces an equality constraint into the model prior
\(
\Theta
=
\Bigl\{
\mathrm{vec}(\theta)
=
c_{\theta}
+
G_{\theta,d}\alpha_\theta
+
G_{\theta,s}\nu_\theta
:\;
\alpha_\theta\in[-1,1]^{m_{\theta_d}},
\;
\nu_\theta\sim\mathcal N(0,I_{m_{\theta_s}})
\Bigr\}.
\)
This equality constraint is not discarded. It is carried through
Lemma~\ref{lem:prior_to_wk} and appears in the induced disturbance set
\eqref{eq:wtheta_k_set}. Similarly, the learned aleatory set is used as
the updated disturbance prior in Lemma~\ref{lem:datacons_nzmean}. Hence, the
batch-wise order is
\begin{equation}
\mathcal W^{(j)},\mathcal S_{AB}^{(j)}
\Longrightarrow
\widehat{\mathcal W}^{c}_{k}
\Longrightarrow
\mathcal S_{AB}^{(j+1)}
\Longrightarrow
\mathcal W_{\star}^{(j+1)}
\Longrightarrow
\widehat{\mathcal R}_{k}^{(j+1)} .
\label{eq:sequential_refinement_order}
\end{equation}
Thus, improved priors refine the realized disturbances, and the refined realized disturbances are then used to further tighten the model and disturbance priors batch by batch. Based on Lemma \ref{lem:scenario_optimization_bound}, until the scenario condition
\eqref{eq:scenario_complexity} with \(\rho=\delta\) is satisfied,
Algorithm~\ref{alg:refined_reach} continues to use the conservative aleatory
prior for future propagation, while realized-disturbance and model-set
refinements can still be updated periodically.
\end{remark}
}

\section{Reachable-Set Refinement by Disturbance Reduction}
\label{sec:main_results}
This section leverages the results of the uncertainty reduction section to progressively refine reachable-set over-approximations. Algorithm 2 presents the proposed reachability estimation approach, which achieves substantially tighter reachable-set overapproximations than the conservative
baseline of Algorithm~\ref{alg:hp_reach_single_window}.

% The following algorithm summarizes the proposed procedure for periodically
% refining reachable sets by combining: (i) consistency-based refinement of realized disturbances using measured data and
% prior knowledge, which reduces uncertainty in the realized disturbance sequence and yields a tighter data-consistent
% model set; (ii) a conservative, nested refinement of the aleatory disturbance set using the refined
% realized-disturbance proxies, ensuring the learned disturbance model remains contained in the prior admissible set; and
% (iii) forward propagation with the refined model and disturbance sets to compute a deterministic reachable-set
% over-approximation together with a formal high-probability containment guarantee for the true reachable sets. The
% algorithm updates the reachable set after every fixed number of samples by tightening both disturbance representations,
% thereby reducing conservativeness in prediction and control.

\begin{algorithm}[t]
\caption{Periodic Reachable-Set Refinement via Realized
and Aleatory Disturbance Learning}
\label{alg:refined_reach}
\begin{algorithmic}[1]
\Require initial dataset \(\mathcal D^{(0)}\); initial set \(\mathcal X_0\);
horizon \(H\); input sets \(\{\mathcal U_h\}_{h=0}^{H-1}\);
initial disturbance prior \(\mathcal W^{(0)}\); update period \(M\);
confidence parameter \(\delta\).
\Ensure reachable tubes \(\{\widehat{\mathcal R}_h^{(j)}\}_{h=0}^{H}\)
and updated models \((\mathcal S^{(j)},\mathcal W^{(j)})\).

\State \(t=0,\; j=0,\; \mathcal D\gets\mathcal D^{(0)}\).
\While{true}
\If{\(\operatorname{rank}\!\left(\begin{bmatrix}X_0\\ U_0\end{bmatrix}\right)=n+m\)
\textbf{and} \(\operatorname{mod}(t,M-1)=0\)}
    \State Freeze \(\mathcal W^{(j)}\).
    \State Use Lemma~\ref{lem:datacons_nzmean} with \(\mathcal W^{(j)}\) to obtain the data-consistent stacked proxy \(\widehat{\mathcal W}^{c}_{D,N}\).
    \State Apply Lemma~\ref{lem:generalized_constraint_absorption} to the stacked proxy \(\widehat{\mathcal W}^{c}_{D,N}\) and project column-wise (retaining the global deterministic constraint) to obtain \(\{\widehat{\mathcal W}^{c}_{D,k}\}_{k=0}^{T-1}\).
    \State Use Lemma~\ref{lem:prior_to_wk} to obtain model-consistent proxies \(\{\widehat{\mathcal W}^{c}_{\theta,k}\}_{k=0}^{T-1}\), whose latent variables are shared across \(k\).
    \State Use Lemma \ref{lem:exact_cpz_intersection} to fuse \(\widehat{\mathcal W}^{c}_{D,k}\) with \(\widehat{\mathcal W}^{c}_{\theta,k}\) to obtain \(\widehat{\mathcal W}^{c}_{k}\).
    \State Stack \(\{\widehat{\mathcal W}^{c}_{k}\}_{k=0}^{T-1}\), construct
    \(\mathcal S^{(j+1)}\gets\mathcal S_{AB}(\widehat{\mathcal W}^{c}_{N})\) via Lemma~\ref{lem:AB_from_CMPZ}, and update the model prior.
    \If{the scenario condition in Lemma~\ref{lem:scenario_optimization_bound} holds}
        \State Learn \(\mathcal W_\star\) via Lemma~\ref{lem:exact_containment_set_refinement} and set \(\mathcal W^{(j+1)}\gets\mathcal W_\star\).
    \Else
        \State Set \(\mathcal W^{(j+1)}\gets\mathcal W^{(j)}\).
    \EndIf
    \State \(\widehat{\mathcal R}^{(j)}_0\gets\mathcal X_0\).
    \For{\(h=0,\ldots,H-1\)}
        \State \(\widehat{\mathcal R}^{(j)}_{h+1}
        \gets
        \mathcal S^{(j+1)}(\delta)
        \boxtimes
        (\widehat{\mathcal R}^{(j)}_h\times\mathcal U_h)
        \oplus
        \mathcal W^{(j+1)}(\delta)\).
    \EndFor
    \State \(j\gets j+1,\quad \mathcal D\gets\emptyset\).
\EndIf
\State \(\mathcal D\gets\mathcal D\cup\{(x_t,x_{t+1},u_t)\}\),
\(t\gets t+1\).
\EndWhile
\end{algorithmic}
\end{algorithm}

% Algorithm~\ref{alg:refined_reach} summarizes the proposed periodic
% refinement procedure. As new input--state data become available, the
% algorithm updates the reachable-set computation after every fixed number
% of samples by combining three steps. First, it refines the past realized
% disturbance sequence using Lemma~\ref{lem:prior_to_wk} and
% Lemma~\ref{lem:datacons_nzmean}: Lemma~\ref{lem:prior_to_wk}, combined with Proposition 1, uses prior
% model information and disturbance information to construct prior-consistent realized-disturbance
% proxies, while Lemma~\ref{lem:datacons_nzmean} enforces data consistency constraints over them. These refined realized proxies are then used in Lemma~\ref{lem:AB_from_CMPZ} to obtain a tighter
% data-consistent model set. Second, once enough refined realized-disturbance proxies have been
% collected, Lemma~\ref{lem:exact_containment_set_refinement} performs a
% conservative nested refinement of the aleatory disturbance set, ensuring
% that the learned disturbance model remains contained in the prior
% admissible set. Third, the algorithm propagates the refined model and
% disturbance sets forward using the recursion analyzed below. Thus,
% Algorithm~\ref{alg:refined_reach} reduces conservatism in prediction and
% control by periodically tightening both the model uncertainty induced by
% past data and the disturbance uncertainty used for future propagation.

Algorithm~\ref{alg:refined_reach} follows the sequential update logic in
Remark~\ref{rem5}. At update window \(j\), the current disturbance prior
\(\mathcal W^{(j)}\) is frozen. Lemma~\ref{lem:datacons_nzmean} uses this
frozen prior to impose trajectory consistency and construct data-consistent
realized-disturbance proxies. In parallel, Lemma~\ref{lem:prior_to_wk} uses
the current model prior, which may have been updated from previous learned
model sets as described in Remark~\ref{rem5}, to construct
model-consistent proxies.

The two proxy descriptions are fused using
Lemma~\ref{lem:exact_cpz_intersection}. The fused proxies then tighten the
data-consistent model set through Lemma~\ref{lem:AB_from_CMPZ}. Once enough
refined proxies are available, Lemma~\ref{lem:exact_containment_set_refinement}
contracts the future aleatory disturbance set while keeping it nested inside
the frozen prior. The propagation loop then uses the updated model set and
the refined disturbance set.

This subsection proves the containment guarantees for the final
propagation loop of Algorithm~\ref{alg:refined_reach}. 

% Once the corresponding
% \(\delta\)-confidence sets from Definition~\ref{def:confidence_set} are
% fixed, the reachable-set recursion is deterministic. 

% The next lemma formalizes the one-step containment property used by the propagation loop
% of Algorithm~\ref{alg:refined_reach}, and the subsequent corollary gives
% the corresponding multi-step high-probability tube guarantee.

Fix a confidence level \(\delta\in(0,1)\) as in
Definition~\ref{def:confidence_set}, and set \(\varepsilon=\delta\). Let
\(\widehat{\mathcal{S}}_{AB}^{c}(\delta)\) denote the truncated dynamics
set from Definition~\ref{def:confidence_set} and define the event
\begin{equation}
\label{eq:event_AB_trunc}
E_{AB}(\delta)
= \{[A\ B]\in\widehat{\mathcal{S}}_{AB}^{c}(\delta)\},
\,
\mathbb{P}\bigl(E_{AB}(\delta)\bigr)= 1-\varepsilon .
\end{equation}
Let \(\mathcal{W}_\star(\delta)\) denote the high-confidence refined disturbance zonotope and define the per-step event
\begin{equation}
\label{eq:event_w_trunc}
E_{w,k}(\delta) = \{w_k\in\mathcal{W}_\star(\delta)\},
\,\,
\mathbb{P}\bigl(E_{w,k}(\delta)\bigr)= 1-\varepsilon,
\quad k\ge 0.
\end{equation}
No temporal independence is assumed in~\eqref{eq:event_w_trunc}.

\begin{lemma}
\label{lem:one_step_trunc_invariance_no_indep}
Fix deterministic input sets \(\{\mathcal{U}_k\}_{k\ge 0}\) and
initial set \(\mathcal{X}_0\). Let the truncation events
\(E_{AB}(\delta)\) and \(E_{w,k}(\delta)\) be defined as
in~\eqref{eq:event_AB_trunc}--\eqref{eq:event_w_trunc}.
Define \(\hat{\mathcal{R}}_0=\mathcal{X}_0\) and, for \(k\ge 0\),
\begin{equation}
\label{eq:reach_recursion_MZ}
\hat{\mathcal{R}}_{k+1}
=
\widehat{\mathcal{S}}_{AB}^{c}(\delta)
\boxtimes\bigl(\hat{\mathcal{R}}_k\times\mathcal{U}_k\bigr)
\oplus\mathcal{W}_\star(\delta),
\end{equation}
where \(\boxtimes\) denotes the exact linear map of a set through
all matrices in \(\widehat{\mathcal{S}}_{AB}^{c}(\delta)\), defined
as~\cite{zhang2025data}. Let
\(E_k^\delta
= E_{AB}(\delta)\cap\bigcap_{i=0}^{k-1}E_{w,i}(\delta)\).
Then, for all \(k\ge 0\),
\begin{equation}
\label{eq:cond_prob_one_step}
\mathbb{P}\!\left(
x_{k+1}\in\hat{\mathcal{R}}_{k+1}
\;\big|\;
E_k^\delta\cap E_{w,k}(\delta)
\right)=1,
\end{equation}
equivalently,
\begin{equation}
\label{eq:joint_prob_one_step}
\mathbb{P}\!\left(
x_{k+1}\in\hat{\mathcal{R}}_{k+1}
\wedge E_k^\delta\wedge E_{w,k}(\delta)
\right)
=
\mathbb{P}\!\left(
E_k^\delta\wedge E_{w,k}(\delta)
\right).
\end{equation}
\end{lemma}

\begin{proof}
Fix any \(\omega\in E_k^\delta\cap E_{w,k}(\delta)\), so that
\([A^\star(\omega)\;\;B^\star(\omega)]
\in\widehat{\mathcal{S}}_{AB}^{c}(\delta)\)
and \(w_k(\omega)\in\mathcal{W}_\star(\delta)\).
On \(E_k^\delta\), by induction \(x_k\in\hat{\mathcal{R}}_k\)
(base case \(x_0\in\mathcal{X}_0=\hat{\mathcal{R}}_0\) holds
by definition). For any \(u_k\in\mathcal{U}_k\), the true
dynamics give
\begin{align}
x_{k+1}
=  [A^\star\;\;B^\star]
\begin{bmatrix}x_k\\u_k\end{bmatrix}
+ w_k(\omega).
\end{align}
Since \([A^\star\;\;B^\star]\in\widehat{\mathcal{S}}_{AB}^{c}(\delta)\)
and \((x_k,u_k)\in\hat{\mathcal{R}}_k\times\mathcal{U}_k\), one has
\begin{align}
[A^\star\;\;B^\star]
\begin{bmatrix}x_k\\u_k\end{bmatrix}
\in
\widehat{\mathcal{S}}_{AB}^{c}(\delta)
\boxtimes
\bigl(\hat{\mathcal{R}}_k\times\mathcal{U}_k\bigr).
\end{align}
Adding \(w_k(\omega)\in\mathcal{W}_\star(\delta)\) gives
\begin{align}
x_{k+1}
\in
\widehat{\mathcal{S}}_{AB}^{c}(\delta)
\boxtimes
\bigl(\hat{\mathcal{R}}_k\times\mathcal{U}_k\bigr)
\oplus\mathcal{W}_\star(\delta)
=\hat{\mathcal{R}}_{k+1}.
\end{align}
Since this holds for every
\(\omega\in E_k^\delta\cap E_{w,k}(\delta)\),
equation~\eqref{eq:cond_prob_one_step} follows.
Equation~\eqref{eq:joint_prob_one_step} is the equivalent
joint-probability form, obtained by multiplying both sides
of~\eqref{eq:cond_prob_one_step} by
\(\mathbb{P}(E_k^\delta\cap E_{w,k}(\delta))\).
\end{proof}

The previous lemma is a conditional containment result, and not a finite-horizon probability bound. The next corollary gives an explicit multi-step probability bound for the reachable tube produced by Algorithm~\ref{alg:refined_reach}. Because the high-confidence model and disturbance events are inferred from the same data, we do not assume them independent; instead we use a union (Bonferroni) bound, which holds regardless of their dependence.

\begin{corollary}
\label{thm:reach_hp_prefix_indep}
Consider the system~\eqref{syst} with \(x_0\in\mathcal{X}_0\) and
\(u_k\in\mathcal{U}_k\). Let \(\{\hat{\mathcal{R}}_k\}_{k\ge 0}\) be
generated by Algorithm~\ref{alg:refined_reach} through the exact
multiplication recursion
\begin{equation}
\label{eq:reach_recursion_MZ_corollary}
\hat{\mathcal R}_0=\mathcal X_0,\qquad
\hat{\mathcal R}_{k+1}
=
\widehat{\mathcal S}_{AB}^{c}(\delta)
\boxtimes
\bigl(\hat{\mathcal R}_k\times\mathcal U_k\bigr)
\oplus
\mathcal W_\star(\delta),
\end{equation}
where \(\boxtimes\) denotes exact multiplication of the matrix set with the
state-input set \cite{zhang2025data}, and \(\delta\) denotes the confidence level in
Definition~\ref{def:confidence_set}.  Then, for every finite horizon \(k\ge0\),
\begin{equation}
\label{eq:reach_prob_prefix_indep}
\mathbb{P}\!\left(
\mathcal{R}_t\subseteq\hat{\mathcal{R}}_t,\;
t=0,\ldots,k
\right)
\ge
1-(k+1)\delta.
\end{equation}
\end{corollary}

\begin{proof}
% For each \(k\ge 0\), define
% \begin{equation}
% \label{eq:prefix_event_indep}
% E_k
% =
% E_{AB}(\delta)\cap\bigcap_{t=0}^{k-1}E_{w,t}(\delta),
% \end{equation}
% with \(\bigcap_{t=0}^{-1}(\cdot)\) being the sure event.
% On \(E_k\), the containment follows by the same induction as in
% Theorem~\ref{thm:hp_reach_single_window}. The base case holds because
% \(\hat{\mathcal R}_0=\mathcal X_0\). If
% \(x_t\in\hat{\mathcal R}_t\) for some \(t<k\), then on
% \(E_{AB}(\delta)\), \([A\ B]\in\widehat{\mathcal{S}}_{AB}^{c}(\delta)\),
% and on \(E_{w,t}(\delta)\), \(w_t\in\mathcal{W}_\star(\delta)\). Hence,
% for any \(u_t\in\mathcal U_t\),
% \begin{align}
% x_{t+1}
% &=
% [A\ B]
% \begin{bmatrix}
% x_t\\ u_t
% \end{bmatrix}
% +w_t
% \nonumber\\
% &\in
% \widehat{\mathcal{S}}_{AB}^{c}(\delta)
% \boxtimes
% \bigl(\hat{\mathcal R}_t\times\mathcal U_t\bigr)
% \oplus
% \mathcal{W}_\star(\delta)
% =
% \hat{\mathcal R}_{t+1}.
% \end{align}
% Thus, \(E_k\) implies
% \begin{equation}
% \mathcal R_t\subseteq\hat{\mathcal R}_t,
% \qquad t=0,\ldots,k .
% \end{equation}
% Therefore,
% \begin{equation}
% \mathbb{P}\!\left(
% \mathcal{R}_t\subseteq\hat{\mathcal{R}}_t,\;
% t=0,\ldots,k
% \right)
% \ge
% \mathbb{P}(E_k).
% \end{equation}
% Using the assumed independence,
% \begin{align}
% \mathbb{P}(E_k)
% &=
% \mathbb{P}\!\left(
% E_{AB}(\delta)\cap\bigcap_{t=0}^{k-1}E_{w,t}(\delta)
% \right)
% \nonumber\\
% &=
% \mathbb{P}(E_{AB}(\delta))
% \prod_{t=0}^{k-1}\mathbb{P}(E_{w,t}(\delta))
% \nonumber\\
% &=
% p_\delta^{\,k+1}
% =
% (1-\varepsilon_\delta)^{k+1}.
% \end{align}
% Combining the last two inequalities gives
% \eqref{eq:reach_prob_prefix_indep}.
By the induction in Theorem~\ref{thm:hp_reach_single_window}, the containment
\(\mathcal R_t\subseteq\hat{\mathcal R}_t\) for all \(t=0,\ldots,k\) holds on the
intersection of the high-confidence model event \(E_{AB}(\delta)\) and the per-step
disturbance events \(\{E_{w,t}(\delta)\}_{t=0}^{k-1}\). Each of these \(k+1\) events
has probability at least \(1-\delta\) by Definition~\ref{def:confidence_set}, i.e.\
fails with probability at most \(\delta\). By the union (Bonferroni) bound, the
probability that any of them fails is at most \((k+1)\delta\), without any
independence assumption between the events. Hence
\(\mathbb{P}\!\left(\mathcal R_t\subseteq\hat{\mathcal R}_t,\;t=0,\ldots,k\right)\ge 1-(k+1)\delta\).
\end{proof}

% \begin{remark}
% The recursion in Algorithm~\ref{alg:refined_reach} is deterministic once
% \(\mathcal S^{(j)}(\delta)\) and \(\mathcal W^{(j)}(\delta)\) are fixed.
% Thus, all probability statements come from the confidence events for the
% model and disturbance sets. The realized-disturbance lemmas improve the
% model set used in propagation, while
% Lemma~\ref{lem:exact_containment_set_refinement} improves the aleatory
% disturbance set used for future disturbances. These two refinements enter
% Algorithm~\ref{alg:refined_reach} at different stages but jointly reduce
% the conservatism of the reachable tube.
% \end{remark}

\begin{proposition}
\label{prop:alg2_less_conservative}
Consider the same initial set \(\mathcal X_0\), input sets
\(\{\mathcal U_k\}_{k\ge0}\), and confidence level \(\delta\) in
Algorithms~\ref{alg:hp_reach_single_window} and~\ref{alg:refined_reach}.
Let \(\mathcal{S}_{AB}^{\mathrm{init}}(\delta)\) and
\(\mathcal W^{\mathrm{init}}(\delta)\) denote the model and aleatory
disturbance sets used by Algorithm~\ref{alg:hp_reach_single_window}. Let
\(\widehat{\mathcal S}_{AB}^{c}(\delta)\) and
\(\mathcal W_\star(\delta)\) denote the refined model and aleatory
disturbance sets used by the constrained branch of
Algorithm~\ref{alg:refined_reach}. Then,
\begin{align}
\widehat{\mathcal S}_{AB}^{c}(\delta)
\subseteq
\mathcal{S}_{AB}^{\mathrm{init}}(\delta),
\qquad
\mathcal W_\star(\delta)
\subseteq
\mathcal W^{\mathrm{init}}(\delta).
\label{eq:alg2_smaller_model_disturbance}
\end{align}
Consequently, the reachable tube generated by
Algorithm~\ref{alg:refined_reach} is contained in the tube generated by
Algorithm~\ref{alg:hp_reach_single_window}, i.e.,
\begin{align}
\widehat{\mathcal R}^{\,\mathrm{ref}}_k
\subseteq
\widehat{\mathcal R}^{\,\mathrm{init}}_k,
\qquad k\ge0 .
\label{eq:alg2_reach_smaller}
\end{align}
Moreover, as more data become available and additional update windows are
processed, Algorithm~\ref{alg:refined_reach} progressively reduces
conservatism.
\end{proposition}

\begin{proof}
By Lemma~\ref{lem:prior_to_wk}, each realized disturbance is first restricted
using prior system knowledge while remaining consistent with the fixed prior
disturbance description. Lemma~\ref{lem:datacons_nzmean} then imposes the
data-consistency constraints on the same realized disturbance sequence. Hence
the refined constrained realized-disturbance proxy
\(\widehat{\mathcal W}^{c}\) used in Algorithm~\ref{alg:refined_reach}
removes disturbance sequences that are inconsistent with either the prior
model information or the measured data, and therefore satisfies \vspace{-12pt}
\begin{align}
\widehat{\mathcal W}^{c}(\delta)
\subseteq
\mathcal W^{\mathrm{init}}_N(\delta),
\label{eq:refined_disturbance_subset_initial}
\end{align}
where \(\mathcal W^{\mathrm{init}}_N(\delta)\) is the stacked disturbance
set used to construct \(\mathcal{S}_{AB}^{\mathrm{init}}(\delta)\).

By Lemma~\ref{lem:AB_from_CMPZ}, the model set is obtained by the affine map
\(\mathcal S_{AB}(\mathcal W)=\{(X_1-W)D_0^\perp:W\in\mathcal W\}\).
Applying this map to the inclusion
\eqref{eq:refined_disturbance_subset_initial} gives
\begin{align}
\widehat{\mathcal S}_{AB}^{c}(\delta)
=
\mathcal S_{AB}(\widehat{\mathcal W}^{c}(\delta))
\subseteq
\mathcal S_{AB}(\mathcal W^{\mathrm{init}}_N(\delta))
=
\mathcal S_{AB}^{\mathrm{init}}(\delta).
\label{eq:refined_model_subset_initial}
\end{align}
Moreover, Lemma~\ref{lem:exact_containment_set_refinement} gives the nested
aleatory refinement
\begin{align}
\mathcal W_\star(\delta)
\subseteq
\mathcal W^{\mathrm{init}}(\delta).
\label{eq:refined_aleatory_subset_initial}
\end{align}
Together, \eqref{eq:refined_model_subset_initial} and
\eqref{eq:refined_aleatory_subset_initial} prove
\eqref{eq:alg2_smaller_model_disturbance}.

It remains to compare the reachable tubes. Both algorithms start from the
same initial set, so
\(\widehat{\mathcal R}^{\,\mathrm{ref}}_0
=
\widehat{\mathcal R}^{\,\mathrm{init}}_0
=
\mathcal X_0\).
Assume
\(\widehat{\mathcal R}^{\,\mathrm{ref}}_k
\subseteq
\widehat{\mathcal R}^{\,\mathrm{init}}_k\).
Using monotonicity of set propagation and Minkowski addition, together with
\eqref{eq:alg2_smaller_model_disturbance}, yields
\begin{align}
\widehat{\mathcal R}^{\,\mathrm{ref}}_{k+1}
&=
\widehat{\mathcal S}_{AB}^{c}(\delta)
\bigl(
\widehat{\mathcal R}^{\,\mathrm{ref}}_k
\times
\mathcal U_k
\bigr)
\oplus
\mathcal W_\star(\delta)
\nonumber\\
&\subseteq
\mathcal S_{AB}^{\mathrm{init}}(\delta)
\bigl(
\widehat{\mathcal R}^{\,\mathrm{init}}_k
\times
\mathcal U_k
\bigr)
\oplus
\mathcal W^{\mathrm{init}}(\delta)
=
\widehat{\mathcal R}^{\,\mathrm{init}}_{k+1}.
\end{align}
The result follows by induction.

Finally, consider two successive update windows. The next window adds new
prior-consistency and data-consistency restrictions to the realized
disturbance proxies, so the feasible disturbance sequences are restricted to
a subset of those allowed before. Applying the affine map
\(\mathcal S_{AB}(\cdot)\) therefore cannot enlarge the induced model set.
Similarly, whenever the scenario condition is satisfied,
Lemma~\ref{lem:exact_containment_set_refinement} constructs the new aleatory
set as a subset of the previous prior disturbance set. Hence both the model
set and the future disturbance set are non-expanding across updates. By the
same monotonicity argument used above, the reachable tubes are also
non-expanding as more data are incorporated; they become strictly smaller
whenever at least one of these set inclusions is strict.
\end{proof}

\begin{figure*}[!t]
    \centering
    \begin{subfigure}[t]{0.32\textwidth}
        \centering
        \includegraphics[width=\linewidth]{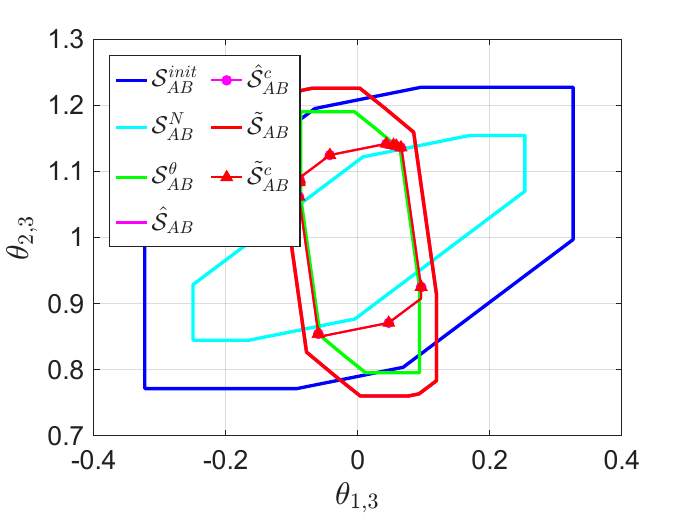}
        \caption{}
        \label{fig:AB_3}
    \end{subfigure}\hfill
    \begin{subfigure}[t]{0.32\textwidth}
        \centering
        \includegraphics[width=\linewidth]{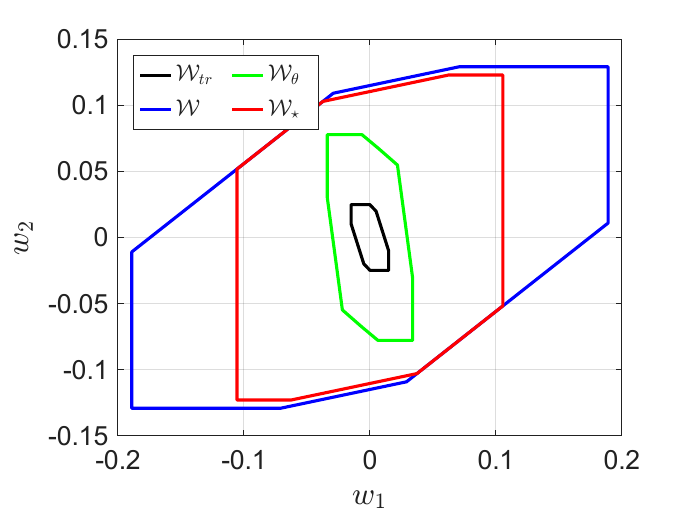}
        \caption{}
        \label{fig:W_al}
    \end{subfigure}\hfill
    \begin{subfigure}[t]{0.32\textwidth}
        \centering
        \includegraphics[width=\linewidth]{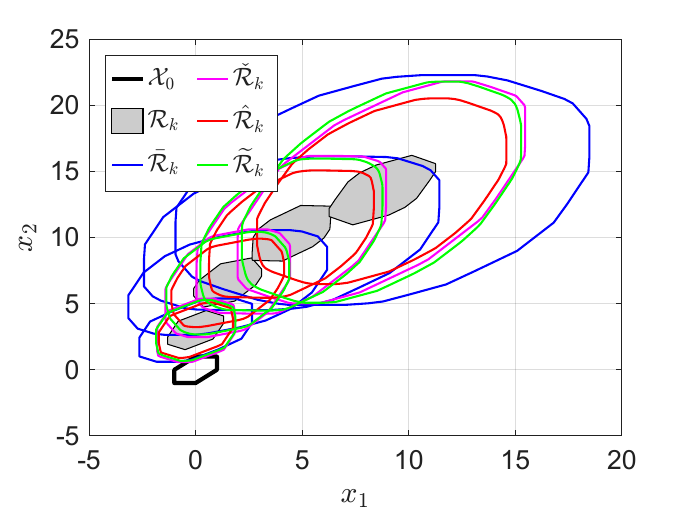}
        \caption{}
        \label{fig:RA}
    \end{subfigure}
    \caption{High-probability sets from data and the resulting reachable sets:
    (a) system-model sets on the input-matrix column
    \([\theta_{1,3},\theta_{2,3}]\); (b) aleatory disturbance sets;
    (c) reachable sets over \(k=1,\dots,4\) (Alanwar baseline omitted).}
    \label{fig:projSetB1}
\end{figure*}

\section{Simulation Results}

% \footnotetext{\href{https://github.com/ZhenZhang32768/Mixed-Zonotope}{https://github.com/ZhenZhang32768/Mixed-Zonotope}}

% \footnote{\url{https://github.com/ZhenZhang32768/Mixed-Zonotope}}

Consider the system \eqref{syst}, where the actual but unknown system dynamics are given by \cite{bisoffi2020data}
\(A=\left[\begin{smallmatrix} 0.8 & 0.5 \\ -0.4 & 1.2 \end{smallmatrix}\right]\),
\(B=\left[\begin{smallmatrix} 0 \\ 1 \end{smallmatrix}\right]\),
and the real disturbance lies within the PZ
\(\mathcal{W}_{\mathrm{tr}}
=\zono{c_{\mathrm{tr}},G_{\mathrm{tr},d},G_{\mathrm{tr},s}}_{\mathcal P}\),
where
\(c_{\mathrm{tr}}=\left[\begin{smallmatrix}0\\0\end{smallmatrix}\right]\),
\(G_{\mathrm{tr},d}=\left[\begin{smallmatrix}-0.0025 & 0.005 \\ 0.0025 & -0.015\end{smallmatrix}\right]\), and
\(G_{\mathrm{tr},s}=\sqrt{0.02}\left[\begin{smallmatrix}0.015 & 0 \\ 0 & 0.015\end{smallmatrix}\right]\).
Here, the prior knowledge is modeled as a probabilistic zonotope, where its
center and generators are formed using data obtained by applying stabilizing
input to the system in \eqref{syst}, but under a different noise set
\(\mathcal{W}_{\theta}
=\zono{c_p,G_{p,d},G_{p,s}}_{\mathcal P}\), given by
\(c_p=\left[\begin{smallmatrix}0\\0\end{smallmatrix}\right]\),
\(G_{p,s}=\sqrt{0.1}\left[\begin{smallmatrix}0.0125 & 0\\0 & 0.0125\end{smallmatrix}\right]\), and
\(G_{p,d}=\left[\begin{smallmatrix}0.006 & -0.0075 & 0.0066 & 0\\-0.0425 & 0.00625 & -0.00525 & 0.01\end{smallmatrix}\right]\).
At initialization, we adopt a deliberately conservative PZ
\(\mathcal{W}=\zono{c,G_d,G_s}_{\mathcal P}\), given by
\(c=\left[\begin{smallmatrix}0\\0\end{smallmatrix}\right]\),
\(G_d=\left[\begin{smallmatrix}0.05 & 0.08\\0.01 & 0.06\end{smallmatrix}\right]\), and
\(G_s=\sqrt{0.02}\left[\begin{smallmatrix}0.12 & 0\\0 & 0.12\end{smallmatrix}\right]\).
In simulation, we set the refinement period to $M=20$ and use a confidence level $\delta=10^{-3}$. Consequently, the high-confidence set contains the true model with probability at least $1-\delta=0.999$.

Fig.~\ref{fig:AB_3} compares several high-probability sets of system models
constructed under the same truncation level, shown on the projection onto the
input-matrix column. Enforcing data consistency (Lemma~\ref{lem:datacons_nzmean}) and propagating the resulting disturbance surrogate $\widehat{\mathcal W}_N$ through the data equation yields the data-consistent model set ${\mathcal{S}}_{AB}^N$, visibly tighter than the initial set ${\mathcal{S}}_{AB}^{init}$; pushing the prior disturbance set \(\mathcal W_{\theta}\) through the data map on the prior window (Lemma~\ref{lem:prior_to_wk}) gives the prior model set \(\mathcal S_{AB}^{\theta}\). These two constrain complementary directions, so the exact CMPZ fusion is the tightest set in this projection (visibly smaller than \emph{both}) and it contains the true system.
% This observation indicates that enforcing data consistency tightens the realized-disturbance description, thereby reducing the conservatism of the disturbance bound and yielding a noticeably tighter system set of system models.

Incorporating the prior (Lemma~\ref{lem:prior_to_wk}) and re-applying Lemma~\ref{lem:datacons_nzmean} gives the refined proxy $\widehat{\mathcal W}^{c}$, from which we build the model set two ways: an MPZ $\hat{\mathcal{S}}_{AB}$ from its unconstrained surrogate $\widehat{\mathcal W}$ (Lemma~\ref{lem:AB_mixed_short}), and a CMPZ $\hat{\mathcal{S}}_{AB}^c$ kept in the equality-coupled representation, so that truncation and propagation respect the coupling constraints. Each is computed for both orders of the underlying proxy intersection (data-then-prior and prior-then-data); the hatted symbols denote one order and the tilded symbols \(\tilde{\mathcal S}_{AB},\tilde{\mathcal S}_{AB}^c\) the other.

% It is important to emphasize that the CPZ-to-PZ transformation is equivalent in the untruncated sense.
% However, under the high-probability truncation adopted in this paper, the two representations are not strictly equivalent. 
% The reason is that PZ truncation effectively imposes an unconstrained truncation on the underlying random generators, whereas the constrained representation enforces equality couplings among random variables, restricting realizable samples to a lower-dimensional feasible subspace. 
% Consequently, even when the same dataset leads to the same proxy intersection $\widehat{\mathcal{W}}^c$, converting it to $\widehat{\mathcal{W}}$ and truncating at the same probability level may produce direction-dependent deterministic zonotope enclosures that differ from those obtained by directly truncating the constrained representation.
CPZ-to-PZ conversion is equivalent only before truncation: since PZ uncertainty
is unbounded, the model sets are confidence enclosures, and high-probability
truncation of the unconstrained surrogate over-approximates the coupled bounded
part. The MPZ sets \(\hat{\mathcal S}_{AB},\tilde{\mathcal S}_{AB}\) therefore
show mild directional expansion relative to the CMPZ sets
\(\hat{\mathcal S}_{AB}^c,\tilde{\mathcal S}_{AB}^c\), a safer enclosure when
different sources constrain different directions, not a violation of refinement.
Building the CMPZ directly from \(\widehat{\mathcal W}^c\) avoids this expansion,
stays inside the unrefined set (Fig.~\ref{fig:AB_3}), and is order-independent
(invariant to the intersection order to machine precision); it is also
noticeably tighter than the deterministic intersection of the same model sets.

Fig.~\ref{fig:projSetB1} separates the two sources of uncertainty that the
framework treats differently. The system-model set in Fig.~\ref{fig:AB_3}
captures \emph{epistemic} uncertainty about the unknown dynamics, which is
reducible with data: enforcing data consistency and prior knowledge shrinks it
substantially. The disturbance set in Fig.~\ref{fig:W_al} captures the
\emph{aleatory} uncertainty of the additive noise, which is inherently
irreducible and, hence, it cannot be contracted below its high-probability envelope at the
prescribed confidence. Accordingly, the learned aleatory set
\(\mathcal W_{\star}\) forms a sound, properly nested sequence
\(\mathcal W_{\mathrm{tr}}\subset\mathcal W_{\theta}\subset\mathcal W_{\star}
\subseteq\mathcal W\): it lies inside the conservative initial set \(\mathcal W\)
and contains the true disturbance, while the prior set \(\mathcal W_{\theta}\) is
itself a strict, informative subset of \(\mathcal W\). The refinement of
\(\mathcal W_{\star}\) is deliberately modest, reflecting that the stochastic
disturbance is genuine noise rather than removable model error; the decisive
conservatism reduction instead comes from the epistemic model-set refinement,
which propagates into the smaller reachable sets of Fig.~\ref{fig:RA}.

The reachable sets are propagated over a four-step horizon from the initial set
\(\mathcal X_0\) with center \(0\) and generator matrix
\(0.5\left[\begin{smallmatrix}1&0&1\\0&1&1\end{smallmatrix}\right]\), under the
input \(\mathcal U=\zono{3,0.25}\).
Table~\ref{tab:results2d} reports, for all methods, the final-step reachable-set
area and the reduction relative to the unrefined data-driven set. The proposed probabilistic CMPZ set \(\hat{\mathcal R}_k\) is the
tightest: it is $55\%$ smaller than the unrefined data-driven set, $48\%$ smaller
than the matrix-zonotope set-membership method of
Alanwar et al.~\cite{alanwar2023data}, and $24\%$ smaller than the deterministic
exact-multiplication method of Zhang et al.~\cite{zhang2025data}, while remaining
larger than, and containing, the model-based set built from the true dynamics.
Both baselines use the same exact set product: Alanwar et al.\ identifies an
unconstrained matrix-zonotope model set, whereas Zhang et al.~\cite{zhang2025data}
takes the geometric intersection of the data-driven and prior model sets together
with the original conservative disturbance \(\mathcal W\), which is neither
probabilistically refined. The gain over Zhang therefore measures the combined
benefit of the probabilistic model-set and disturbance refinement; the disturbance
contribution alone is isolated in Fig.~\ref{fig:noise_adv} below. The exact multiplication is
the costlier propagation (about $1.8$\,s per step against $3$--$18$\,ms for
the zonotope-based baselines), which is the price of the non-convex exact set product paid for the tighter sets.

We validate the high-probability guarantee empirically. Over \(K=10^{6}\)
Monte-Carlo trajectories with the disturbance drawn from the true (untruncated)
distribution and propagated through the true dynamics, the realized state lies in
\(\hat{\mathcal R}_k\) at every step (and jointly over the horizon) with an
empirical frequency whose $95\%$ Clopper--Pearson interval equals
\([1.0000,1.0000]\) to four decimals, for every method, confirming soundness. The
set-based tube is one-sided conservative: Minkowski-sum propagation over-covers,
so the empirical frequency saturates at one. That the guarantee is nonetheless \emph{tight}, not vacuously conservative, is
shown in Fig.~\ref{fig:calib}: sweeping \(\delta\) with each disturbance sampled
from its own distribution, the empirical coverage of the zonotopic \(1-\delta\)
enclosure of \(\mathcal W\) tracks the nominal \(1-\delta\) along the diagonal
for a Gaussian-dominant disturbance; the residual over-coverage comes from the
deterministically-bounded part, contained with certainty.

Fig.~\ref{fig:nsweep} reports a data-convergence study: as the data length \(N\)
grows, the data-driven reachable set contracts and its variability across random
data records (error bars over $10$ seeds) collapses, converging to a floor set by
the irreducible aleatory disturbance, confirming that more data consistently
tighten the model, consistent with the epistemic/aleatory separation above.

On this disturbance, the advantage of the probabilistic over the deterministic
treatment is modest at the nominal noise level but grows with the Gaussian
fraction of the disturbance. Fig.~\ref{fig:noise_adv} isolates this effect: at the
same empirically validated \(1-\delta\) coverage, the probabilistic reachable set
is \(2\)--\(3\) times smaller than the deterministic worst-case enclosure that a
bounded-noise method must use for the unbounded Gaussian component, and the gap
widens with the stochastic fraction. Sweeping \(\delta\) instead of \(s\)
confirms that the confidence level provides an explicit trade-off between set
size and confidence: the four-step area grows monotonically with \(1-\delta\)
and remains well below the deterministic worst-case at every level, with
empirical coverage \(\geq 1-\delta\) throughout.

This is a difference of disturbance \emph{semantics}, not of set arithmetic. Note that the deterministic baseline uses the same exact multiplication. Lacking a confidence
parameter, a bounded-noise method faces a dichotomy against unbounded noise: a
disturbance set matched to the \(1-\delta\) enclosure has its premise
\(w_k\in\mathcal W_{\det}\) violated at rate \(p=3.7\times10^{-4}\) per step
(\(s{=}16\)), so containment fails with probability \(1-(1-p)^{h}\) (about
\(4\%\) at \(h{=}100\), \(31\%\) at \(h{=}1000\)); inflating the set until
violations vanish (\(\delta_{\det}=10^{-9}\)) instead incurs the
\(2\)--\(3\times\) conservatism of Fig.~\ref{fig:noise_adv}. Determinizing the
model set as well (geometrically intersecting the data and prior model sets,
with the refined disturbance bounded at the worst-case level) yields a
still-looser tube (\(142.4\) vs.\ \(115.2\) at \(k{=}4\)); the prior deterministic
method, which instead keeps the original conservative disturbance, is looser yet
(\(151.6\), Table~\ref{tab:results2d}). Only the proposed framework
attaches a valid confidence statement, selecting its operating point explicitly
through \(\delta\).

We deliberately restrict the simulation to a two-state example, which suffices
to illustrate every component of the proposed theory while keeping all sets
directly visualizable in the parameter and state spaces. The scalability of the
underlying set-based machinery to higher-dimensional systems is already
established in prior work: matrix-zonotope data-driven reachability on a
five-state benchmark in~\cite{alanwar2023data}, and the exact-multiplication
propagation employed here on the same five-state benchmark
in~\cite{zhang2025data}.

\subsection{Hardware Validation}
We further validate the framework on a physical \mbox{1/10}-scale car-like
robot (Fig.~\ref{fig:jetracer}). Reflective markers provide ground-truth planar
pose \(x=[p_x,\,p_y,\,\psi]^\top\) from an external motion-capture system, and
the robot is commanded through its velocity interface with input
\(u=[v,\,\omega]^\top\) (linear and angular velocity). Input--output
trajectories were logged at \(20\,\mathrm{Hz}\) and resampled onto a
\(10\,\mathrm{Hz}\) grid (sample time \(0.1\,\mathrm{s}\)) for one-step
identification and reachability; the high-probability model and
disturbance sets were identified from three pooled excitation runs
(\(T=261\) one-step samples, with \([X_0^\top\;U_0^\top]^\top\) of full rank
\(n+m=5\)), and the \(1-\delta\) containment guarantee was tested on an
independent held-out run, with unmodeled kinematics and sensor noise absorbed
into the probabilistic disturbance set \(W\). Fig.~\ref{fig:hwreach} reports the
result: the data-driven probabilistic reachable sets track the held-out
trajectory, and the empirical held-out coverage is \(96.4\%\), \(98.2\%\), and
\(98.2\%\) at the target levels \(1-\delta\in\{90\%,95\%,99\%\}\),i.e., the
containment guarantee is met on real data. Pooling several runs is needed so
that the estimated disturbance set captures the run-to-run variation; a
single-run estimate is over-optimistic and does not generalize.

\begin{figure}[!t]
    \centering
    \includegraphics[width=0.5\columnwidth]{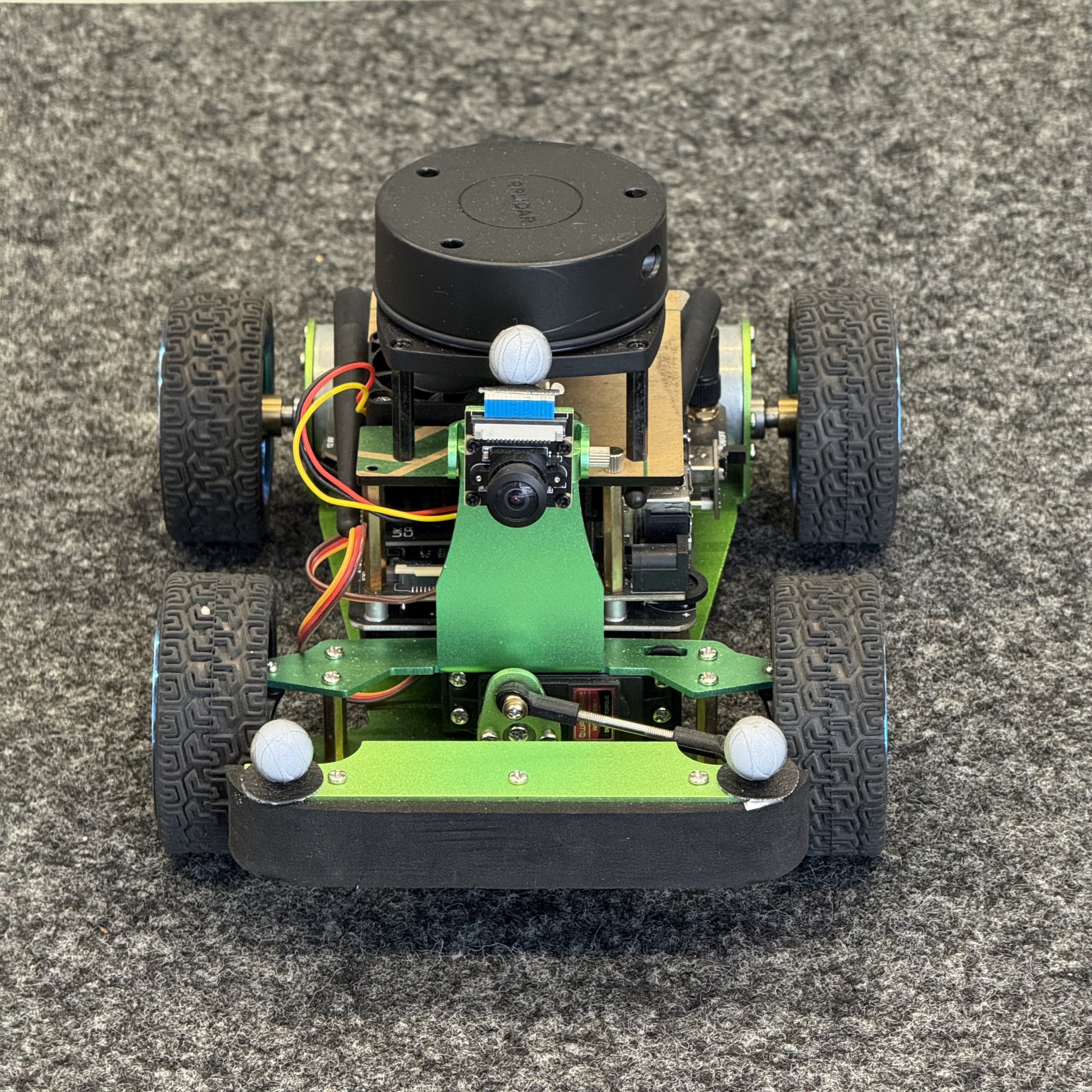}
    \caption{Car-like robot used for the hardware validation, with
    motion-capture markers for ground-truth pose.}
    \label{fig:jetracer}
\end{figure}

\begin{figure*}[!t]
    \centering
    \begin{subfigure}[t]{0.32\textwidth}
        \centering
        \includegraphics[width=\linewidth]{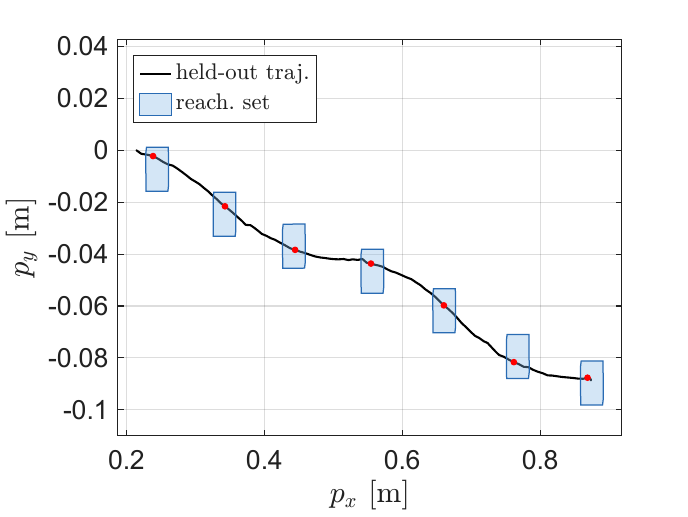}
        \caption{}
        \label{fig:hw_a}
    \end{subfigure}\hfill
    \begin{subfigure}[t]{0.32\textwidth}
        \centering
        \includegraphics[width=\linewidth]{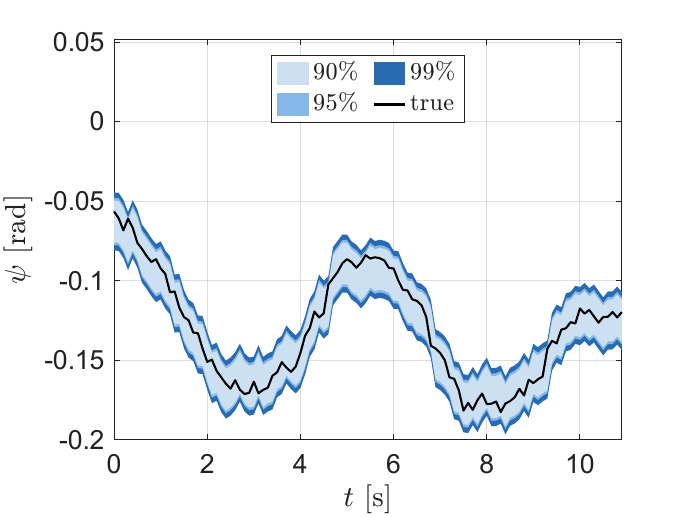}
        \caption{}
        \label{fig:hw_b}
    \end{subfigure}\hfill
    \begin{subfigure}[t]{0.32\textwidth}
        \centering
        \includegraphics[width=\linewidth]{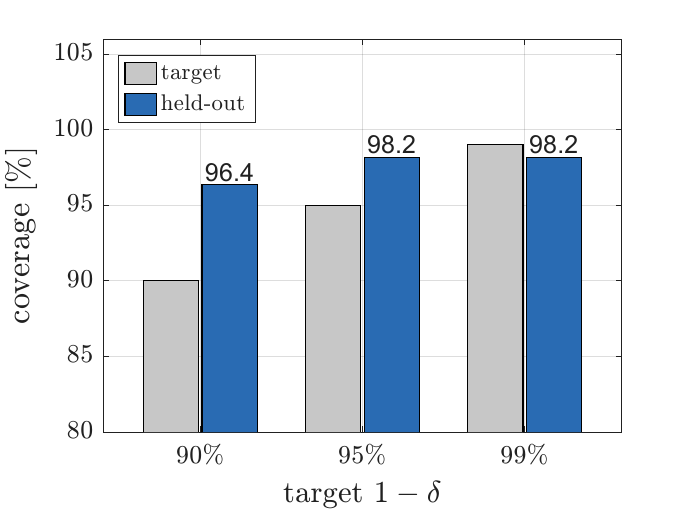}
        \caption{}
        \label{fig:hw_c}
    \end{subfigure}
    \caption{Hardware validation on the car-like robot:
    (a)~data-driven probabilistic reachable sets (blue) along a held-out
    trajectory (black, top view); (b)~measured heading \(\psi\) within the
    \(1-\delta\) band; (c)~empirical held-out coverage versus the target
    \(1-\delta\) at three confidence levels.}
    \label{fig:hwreach}
\end{figure*}

\begin{table}[!t]
    \centering
    \caption{Final-step (\(k=4\)) reachable-set area and reduction relative to
    the unrefined data-driven set.}
    \label{tab:results2d}
    \footnotesize
    \setlength{\tabcolsep}{4pt}
    \begin{tabular}{lcc}
    \toprule
    Reachable set & area & reduction \\
    \midrule
    Model-based (true \([A\,B]\)) & \(15.6\) & \(93.9\%\) \\
    Alanwar et al.~\cite{alanwar2023data} & \(221.4\) & \(13.4\%\) \\
    Unrefined data-driven & \(255.5\) & \(0.0\%\) \\
    Refined (MPZ) & \(148.1\) & \(42.1\%\) \\
    Proposed (prob., exact mult.) & \(\mathbf{115.2}\) & \(\mathbf{54.9\%}\) \\
    Deterministic worst-case~\cite{zhang2025data}\(^{\dagger}\) & \(151.6\) & \(40.7\%\) \\
    \bottomrule
    \end{tabular}
    \par\smallskip
    \begin{minipage}{\columnwidth}\footnotesize
    \(^{\dagger}\)The prior deterministic method: the model set is the geometric
    intersection of the data-driven and prior model sets, and the disturbance is
    the original conservative set \(\mathcal W\) (neither probabilistically
    refined) bounded at the worst-case level \(\delta_{\det}=10^{-9}\).
    \end{minipage}
\end{table}

\begin{figure*}[!t]
    \centering
    \begin{subfigure}[t]{0.32\textwidth}
        \centering
        \includegraphics[width=\linewidth]{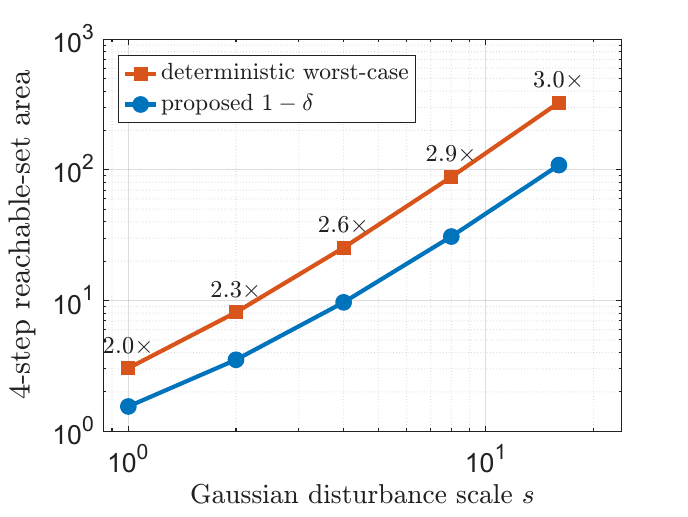}
        \caption{}
        \label{fig:noise_adv}
    \end{subfigure}\hfill
    \begin{subfigure}[t]{0.32\textwidth}
        \centering
        \includegraphics[width=\linewidth]{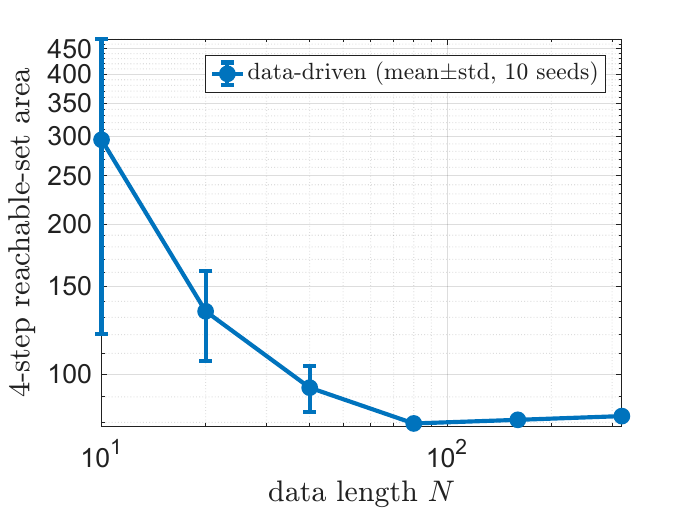}
        \caption{}
        \label{fig:nsweep}
    \end{subfigure}\hfill
    \begin{subfigure}[t]{0.32\textwidth}
        \centering
        \includegraphics[width=\linewidth]{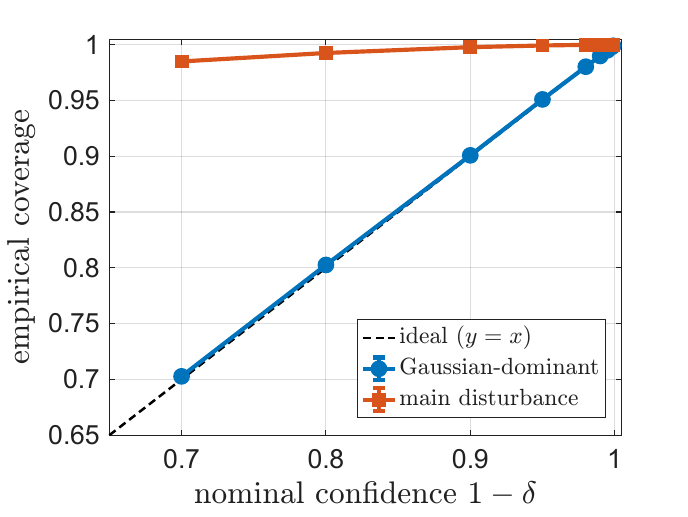}
        \caption{}
        \label{fig:calib}
    \end{subfigure}
    \caption{Quantitative studies: (a) four-step reachable-set area versus the
    Gaussian disturbance scale \(s\) (probabilistic vs.\ deterministic
    worst-case); (b) four-step area versus data length \(N\) (mean \(\pm\) std,
    $10$ records); (c) empirical coverage of the \(1-\delta\) disturbance
    enclosure versus nominal confidence (\(95\%\) Clopper--Pearson).}
\end{figure*}

\section{Conclusion}

This paper developed a data-driven reachability framework for discrete-time
linear systems with additive disturbances modeled by probabilistic zonotopes.
The proposed approach separates past realized disturbances from future
aleatory disturbances and refines them at different stages of the pipeline.
Realized disturbances are tightened using prior model information and
data-consistency constraints, yielding a less conservative data-consistent
model set. Once enough refined proxies are collected, the conservative
aleatory disturbance set is contracted through a scenario-supported refinement
program and used for future propagation. The resulting reachable tubes preserve
high-probability containment guarantees while reducing conservatism at the
same confidence level. Numerical results demonstrate tighter model,
disturbance, and reachable sets compared with unrefined probabilistic and
deterministic baselines. Future work will extend the framework to nonlinear
systems.

\bibliographystyle{IEEEtran} % 
\bibliography{refs}          % 

@InProceedings{UncertaintyLearning1,
  title     = {Learning-based Rigid Tube Model Predictive Control},
  author    = {Yulong Gao and Shuhao Yan and Jian Zhou and Mark Cannon and Alessandro Abate and Karl Henrik Johansson},
  booktitle = {Proceedings of the 6th Annual Learning for Dynamics \& Control Conference},
  series    = {Proceedings of Machine Learning Research},
  volume    = {242},
  pages     = {492--503},
  year      = {2024},
}

@ARTICLE{1632303,
  author={Calafiore, G.C. and Campi, M.C.},
  journal={IEEE Transactions on Automatic Control}, 
  title={The scenario approach to robust control design}, 
  year={2006},
  volume={51},
  number={5},
  pages={742-753},
  }

@article{UncertaintyLearning2,
  author    = {Julian Berberich and Carsten W. Scherer and Frank Allgower},
  title     = {Combining Prior Knowledge and Data for Robust Controller Design},
  journal   = {IEEE Transactions on Automatic Control},
  volume    = {68},
  number    = {8},
  pages     = {4618--4633},
  year      = {2022},
}

@article{ganai2024hamilton,
  title={Hamilton-jacobi reachability in reinforcement learning: A survey},
  author={Ganai, Milan and Gao, Sicun and Herbert, Sylvia L},
  journal={IEEE Open Journal of Control Systems},
  volume={3},
  pages={310--324},
  year={2024},
  publisher={IEEE}
}

@article{althoff2021set,
  title={Set propagation techniques for reachability analysis},
  author={Althoff, Matthias and Frehse, Goran and Girard, Antoine},
  journal={Annual Review of Control, Robotics, and Autonomous Systems},
  volume={4},
  number={1},
  pages={369--395},
  year={2021},
  publisher={Annual Reviews}
}

@inproceedings{althoff2023checking,
  title={Checking and establishing reachset conformance in CORA 2023},
  author={Althoff, Matthias},
  booktitle={Proc. of 10th International Workshop on Applied Verification of Continuous and Hybrid Systems},
  year={2023}
}

@article{wang2025system,
  title={System Identification in the Network Era: A Survey of Data Issues and Innovative Approaches},
  author={Wang, Qing-Guo and Zhang, Liang},
  journal={IEEE/CAA Journal of Automatica Sinica},
  volume={12},
  number={7},
  pages={1305--1319},
  year={2025},
  publisher={IEEE}
}

@inproceedings{luo2023reachability,
  title={Reachability analysis for linear systems with uncertain parameters using polynomial zonotopes},
  author={Luo, Ertai and Kochdumper, Niklas and Bak, Stanley},
  booktitle={Proceedings of the 26th ACM International Conference on Hybrid Systems: Computation and Control},
  pages={1--12},
  year={2023}
}

@inproceedings{doshi2022hamilton,
  title={Hamilton-Jacobi multi-time reachability},
  author={Doshi, Manan and Bhabra, Manmeet and Wiggert, Marius and Tomlin, Claire J and Lermusiaux, Pierre FJ},
  booktitle={2022 IEEE 61st Conference on Decision and Control (CDC)},
  pages={2443--2450},
  year={2022},
  organization={IEEE}
}

@inproceedings{UncertaintyLearning3,
author = {Ashoori, MohammadHossein and Aminzadeh, Ali and Lavaei, Abolfazl and Nejati, Amy},
title = {Physics-Informed Safety Verification of Nonlinear Systems: A Scenario Approach with Data Mitigation},
year = {2025},
publisher = {Association for Computing Machinery},
booktitle = {Proceedings of the ACM/IEEE 16th International Conference on Cyber-Physical Systems (with CPS-IoT Week 2025)},
}

@article{UncertaintyLearning4,
  author    = {Julian Berberich and Anne Romer and Carsten W. Scherer and Frank Allgower},
  title     = {Robust data-driven state-feedback design},
  journal   = {arXiv:1909.04314},
  year      = {2019},
}

@ARTICLE{UncertaintyLearning5,
  author={Modares, Amir and Ghiasi, Niyousha and Kiumarsi, Bahare and Modares, Hamidreza},
  journal={IEEE Transactions on Automatic Control}, 
  title={Unifying Direct and Indirect Learning for Safe Control of Linear Systems}, 
  year={2026},
  volume={},
  number={},
  pages={1-8},
  }

@article{bisoffi2020data,
  title={Data-based guarantees of set invariance properties},
  author={Bisoffi, Andrea and De Persis, Claudio and Tesi, Pietro},
  journal={IFAC-PapersOnLine},
  volume={53},
  number={2},
  pages={3953--3958},
  year={2020},
  publisher={Elsevier}
}

@phdthesis{althoff2010reachability,
  title={Reachability analysis and its application to the safety assessment of autonomous cars},
  author={Althoff, Matthias},
  year={2010},
  school={Technische Universität München}
}

@inproceedings{althoff2009safety,
  title={Safety assessment for stochastic linear systems using enclosing hulls of probability density functions},
  author={Althoff, Matthias and Stursberg, Olaf and Buss, Martin},
  booktitle={2009 European Control Conference (ECC)},
  pages={625--630},
  year={2009},
  organization={IEEE}
}

@article{kuhn1998rigorously,
  title={Rigorously computed orbits of dynamical systems without the wrapping effect},
  author={K{\"u}hn, Wolfgang},
  journal={Computing},
  volume={61},
  number={1},
  pages={47--67},
  year={1998},
  publisher={Springer}
}

@inproceedings{zhang2025data,
  title={Data-driven nonconvex reachability analysis using exact multiplication},
  author={Zhang, Zhen and Niazi, M Umar B and Chong, Michelle S and Johansson, Karl H and Alanwar, Amr},
  booktitle={2025 IEEE 64th Conference on Decision and Control (CDC)},
  pages={4882--4889},
  year={2025},
  organization={IEEE}
}

@book{anderson2005optimal,
  title={Optimal filtering},
  author={Anderson, Brian DO and Moore, John B},
  year={2005},
  publisher={Courier Corporation}
}

@article{mitchell2005time,
  title={A time-dependent Hamilton-Jacobi formulation of reachable sets for continuous dynamic games},
  author={Mitchell, Ian M and Bayen, Alexandre M and Tomlin, Claire J},
  journal={IEEE Transactions on automatic control},
  volume={50},
  number={7},
  pages={947--957},
  year={2005},
  publisher={IEEE}
}

@article{xue2023reach,
  title={Reach-avoid analysis for polynomial stochastic differential equations},
  author={Xue, Bai and Zhan, Naijun and Fr{\"a}nzle, Martin},
  journal={IEEE Transactions on Automatic Control},
  volume={69},
  number={3},
  pages={1882--1889},
  year={2023},
  publisher={IEEE}
}

@article{sieber2022system,
  title={System level disturbance reachable sets and their application to tube-based MPC},
  author={Sieber, Jerome and Zanelli, Andrea and Bennani, Samir and Zeilinger, Melanie N},
  journal={European Journal of Control},
  volume={68},
  pages={100680},
  year={2022},
  publisher={Elsevier}
}

@article{zhang2022robust,
  title={Robust tube-based model predictive control with Koopman operators},
  author={Zhang, Xinglong and Pan, Wei and Scattolini, Riccardo and Yu, Shuyou and Xu, Xin},
  journal={Automatica},
  volume={137},
  pages={110114},
  year={2022},
  publisher={Elsevier}
}

@article{xiang2020reachable,
  title={Reachable set estimation for neural network control systems: A simulation-guided approach},
  author={Xiang, Weiming and Tran, Hoang-Dung and Yang, Xiaodong and Johnson, Taylor T},
  journal={IEEE Transactions on Neural Networks and Learning Systems},
  volume={32},
  number={5},
  pages={1821--1830},
  year={2020},
  publisher={IEEE}
}

@inproceedings{chen2022reachability,
  title={Reachability analysis for cyber-physical systems: Are we there yet?},
  author={Chen, Xin and Sankaranarayanan, Sriram},
  booktitle={NASA formal methods symposium},
  pages={109--130},
  year={2022},
  organization={Springer}
}

@article{wang2023safe,
  title={Safe reinforcement learning for automated vehicles via online reachability analysis},
  author={Wang, Xiao and Althoff, Matthias},
  journal={IEEE Transactions on Intelligent Vehicles},
  year={2023},
  publisher={IEEE}
}

@article{chen2018decomposition,
  title={Decomposition of reachable sets and tubes for a class of nonlinear systems},
  author={Chen, Mo and Herbert, Sylvia L and Vashishtha, Mahesh S and Bansal, Somil and Tomlin, Claire J},
  journal={IEEE Transactions on Automatic Control},
  volume={63},
  number={11},
  pages={3675--3688},
  year={2018},
  publisher={IEEE}
}

@article{alanwar2023data,
  title={Data-driven reachability analysis from noisy data},
  author={Alanwar, Amr and Koch, Anne and Allg{\"o}wer, Frank and Johansson, Karl Henrik},
  journal={IEEE Transactions on Automatic Control},
  volume={68},
  number={5},
  pages={3054--3069},
  year={2023},
  publisher={IEEE}
}

@article{hu2025robust,
  title={Robust Data-Driven Predictive Control for Unknown Linear Systems with Bounded Disturbances},
  author={Hu, Kaijian and Liu, Tao},
  journal={IEEE Transactions on Automatic Control},
  year={2025},
  publisher={IEEE}
}

@inproceedings{gao2022robust,
  title={Robust risk-aware model predictive control of linear systems with bounded disturbances},
  author={Gao, Yulong and Liu, Changxin and Johansson, Karl H},
  booktitle={2022 IEEE 61st Conference on Decision and Control (CDC)},
  pages={1148--1155},
  year={2022},
  organization={IEEE}
}

@article{scott2016constrained,
  title={Constrained zonotopes: A new tool for set-based estimation and fault detection},
  author={Scott, Joseph K and Raimondo, Davide M and Marseglia, Giuseppe Roberto and Braatz, Richard D},
  journal={Automatica},
  volume={69},
  pages={126--136},
  year={2016},
  publisher={Elsevier}
}

@article{raghuraman2022set,
  title={Set operations and order reductions for constrained zonotopes},
  author={Raghuraman, Vignesh and Koeln, Justin P},
  journal={Automatica},
  volume={139},
  pages={110204},
  year={2022},
  publisher={Elsevier}
}

@inproceedings{girard2005reachability,
  title={Reachability of uncertain linear systems using zonotopes},
  author={Girard, Antoine},
  booktitle={International workshop on hybrid systems: Computation and control},
  pages={291--305},
  year={2005},
  organization={Springer}
}

@article{blanchini1999set,
  title={Set invariance in control},
  author={Blanchini, Franco},
  journal={Automatica},
  volume={35},
  number={11},
  pages={1747--1767},
  year={1999},
  publisher={Elsevier}
}

@article{der2009aleatory,
  title={Aleatory or epistemic? Does it matter?},
  author={Der Kiureghian, Armen and Ditlevsen, Ove},
  journal={Structural safety},
  volume={31},
  number={2},
  pages={105--112},
  year={2009},
  publisher={Elsevier}
}

@article{li2025aleatory,
  title={Aleatory and epistemic uncertainty in reliability analysis: An engineering perspective},
  author={Li, Pei-Pei and Valdebenito, Marcos A and Dang, Chao and Beer, Michael and Faes, Matthias GR},
  journal={Structural Safety},
  pages={102666},
  year={2025},
  publisher={Elsevier}
}

@article{chen2024robust,
  title={Robust model predictive control with polytopic model uncertainty through system level synthesis},
  author={Chen, Shaoru and Preciado, Victor M and Morari, Manfred and Matni, Nikolai},
  journal={Automatica},
  volume={162},
  pages={111431},
  year={2024},
  publisher={Elsevier}
}

@article{fiacchini2021probabilistic,
  title={Probabilistic reachable and invariant sets for linear systems with correlated disturbance},
  author={Fiacchini, Mirko and Alamo, Teodoro},
  journal={Automatica},
  volume={132},
  pages={109808},
  year={2021},
  publisher={Elsevier}
}

@article{hewing2019scenario,
  title={Scenario-based probabilistic reachable sets for recursively feasible stochastic model predictive control},
  author={Hewing, Lukas and Zeilinger, Melanie N},
  journal={IEEE Control Systems Letters},
  volume={4},
  number={2},
  pages={450--455},
  year={2019},
  publisher={IEEE}
}

@article{hewing2020recursively,
  title={Recursively feasible stochastic model predictive control using indirect feedback},
  author={Hewing, Lukas and Wabersich, Kim P and Zeilinger, Melanie N},
  journal={Automatica},
  volume={119},
  pages={109095},
  year={2020},
  publisher={Elsevier}
}

@article{li2022set,
  title={Set-based state estimation with probabilistic consistency guarantee under epistemic uncertainty},
  author={Li, Shen and Stouraitis, Theodoros and Gienger, Michael and Vijayakumar, Sethu and Shah, Julie A},
  journal={IEEE Robotics and Automation Letters},
  volume={7},
  number={3},
  pages={5958--5965},
  year={2022},
  publisher={IEEE}
}

@article{shetty2020predicting,
  title={Predicting state uncertainty bounds using non-linear stochastic reachability analysis for urban GNSS-based UAS navigation},
  author={Shetty, Akshay and Gao, Grace Xingxin},
  journal={IEEE Transactions on Intelligent Transportation Systems},
  volume={22},
  number={9},
  pages={5952--5961},
  year={2020},
  publisher={IEEE}
}

\begin{IEEEbiographynophoto}{Amir Modares}
He is a Student Member of IEEE. He received the B.S. degree in
electrical engineering from Azad University, Qaen, Iran, in 2020, and
the M.S. degree in electrical engineering from the Sharif University of
Technology, Tehran, Iran, in 2022. He is currently a Ph.D. student at
the University of Cyprus. His research interests include machine learning,
convex optimization, safe control, and nonlinear control.
\end{IEEEbiographynophoto}

\begin{IEEEbiographynophoto}{Zhen Zhang}
He received the B.Sc.\ degree in Automation from Chang'an University, Xi'an, China, in 2021, and the M.Sc.\ degree in Control Science and Engineering from Northwestern Polytechnical University, Xi'an, China, in 2024. He is currently pursuing the Ph.D.\ degree at the School of Computation, Information and Technology, Technical University of Munich, Germany. His research interests include data-driven reachability analysis, set-based methods, and safety verification of dynamical systems.
\end{IEEEbiographynophoto}

\begin{IEEEbiographynophoto}{Themistoklis Charalambous} (Senior Member, IEEE) received his B.A. and M.Eng. degrees in Electrical and Information Sciences from Trinity College, University of Cambridge. He completed his Ph.D. in the Control Laboratory of the Department of Engineering at the University of Cambridge. Following his Ph.D., he held postdoctoral positions at Imperial College London, the Royal Institute of Technology (KTH), and Chalmers University of Technology. In 2017, he joined Aalto University as a tenure-track Assistant Professor. In 2018, he was awarded an Academy of Finland Research Fellowship, and in 2020, he was appointed tenured Associate Professor at Aalto University. In 2021, he joined the University of Cyprus as a tenure-track Assistant Professor, while remaining affiliated with Aalto University as a Visiting Professor. Since April 2023, he has also been a Visiting Professor at the FinEst Centre for Smart Cities. Since May 2026, he has been a tenured Associate Professor at the University of Cyprus.
\end{IEEEbiographynophoto}

\begin{IEEEbiographynophoto}{Amr Alanwar} He is an assistant professor at Technical University of Munich. He received an M.Sc. in Computer Engineering from Ain Shams University, Cairo, Egypt, in 2013 and a Ph.D. in Computer Science from the Technical University of Munich, Germany, in 2020. He was a postdoctoral researcher at KTH Royal Institute of Technology. He was also a research assistant at the University of California, Los Angeles. He received the Emmy Noether Funding from the German Research Foundation in 2025. Also, he received the Best Paper Award in the Systems and Applications Track at HSCC/ICCPS during CPS Week 2026 and the Best Demonstration Paper Award at the IPSN during CPS Week 2017. He was a finalist in the Qualcomm Innovation Fellowship for two consecutive years.
\end{IEEEbiographynophoto}

\begin{IEEEbiographynophoto}{Hamidreza Modares}
He received the B.S. degree from the University of Tehran,
Tehran, Iran, in 2004, the M.S. degree from the Shahrood University of
Technology, Shahrood, Iran, in 2006, and the Ph.D. degree from the
University of Texas at Arlington, Arlington, TX, USA, in 2015, all in
electrical engineering. He is currently an Associate Professor with the
Department of Mechanical Engineering, Michigan State University. Before
joining Michigan State University, he was an Assistant Professor with the
Department of Electrical Engineering, Missouri University of Science and
Technology. His current research interests include reinforcement learning,
safe control, machine learning in control, distributed control of
multi-agent systems, and robotics. He is an Associate Editor of the
\emph{IEEE Transactions on Systems, Man, and Cybernetics: Systems}.
\end{IEEEbiographynophoto}
\end{document}